\documentclass[preprint]{elsarticle}

\journal{Science of Computer Programming}

\usepackage{url}
\usepackage{color}
\usepackage{amssymb}
\usepackage{amsmath}
\usepackage{alltt}
\usepackage{verbatim}
\usepackage{stmaryrd}
\usepackage{mathtools}
\usepackage{wrapfig}
\usepackage{float}

\usepackage{graphicx}

\newcommand{\ans}[1]{}

\newcommand{\vdashInd}{\vdash_{\it NDF}}

\newcommand{\bottom}{\mathop{\perp}}

\newcommand{\ExpS}{\mathsf{ExpStruct}}
\newcommand{\Fun}{\mathsf{Functor}}
\newcommand{\Fixpv}{\mathsf{FixpVar}}
\newcommand{\Alphb}{\mathsf{Alph}}
\newcommand{\Slt}{\mathsf{Slt}}
\newcommand{\AlphbN}{\mathsf{AlphName}}
\newcommand{\SltN}{\mathsf{SltName}}
\newcommand{\Ingred}{\mathsf{Ingredient}}
\newcommand{\IdtrlG}{\textnormal{$\id\triangleleft\mathscr{G}$}}
\newcommand{\BtrlG}{\textnormal{$\Bl\triangleleft\mathscr{G}$}}

\newcommand{\FtrlG}{\textnormal{$\mathscr{F}\triangleleft\mathscr{G}$}}
\newcommand{\GtrlG}{\textnormal{$\mathscr{G}\triangleleft\mathscr{G}$}}

\newcommand{\efr}[1]{{\fbox{$#1$}}}  
\newcommand{\RG}{\textnormal{${\mathcal R}$}}
\newcommand{\RId}{\textnormal{${\mathcal R_{\it id}}$}}

\newcommand{\R}{\textnormal{${\mathcal R}$}}

\newcommand{\G}{\textnormal{${\mathcal G}$}}

\newcommand{\algspec}{\textnormal{${\mathcal E}_{\itG}$}}
\newcommand{\algspecN}{\textnormal{${\mathcal E}$}}
\newcommand{\behspec}{\textnormal{${\mathcal B}_{\itG}$}}

\newcommand{\ct}{\colon}

\newcommand{\cl}[1]{\it{cl}(#1)}

\newcommand{\CIRC}{\textsf{CIRC}}

\newcommand{\code}[1]{{\tt \fontsize{9}{10}\selectfont {#1}}}

\newcommand{\rTrans}[1]{\mathrel{{\overset{#1}{\Rightarrow}}}}

\newcommand{\eps}{\varepsilon}
\newcommand{\itB}{\mathsf{B}}

\newcommand{\itG}{\mathscr{G}}

\newcommand{\itF}{\mathscr{F}}
\newcommand{\itExp}{\mathsf{Exp}}

\DeclareSymbolFont{lasy}{U}{lasy}{m}{n}
\DeclareMathSymbol\myDiamond{\mathord}{lasy}{"33}
\newcommand{\myplus}{\mathbin{\rlap{$\myDiamond$}\hspace*{.01cm}\raisebox{.14ex}{$+$}}}

\usepackage[all,2cell,cmtip,dvips,ps]{xy}
\usepackage[mathscr]{eucal}

\newcommand\Gf{\mathscr G} 
\newcommand\F{\mathscr F}  
\newcommand\Bl{{\mathsf B}} 
\newcommand\id{{\mathsf{Id}}}
\newcommand\pow{\mathscr{P}_{\!\!\omega}}
\newcommand\pf{{\it NDF}}
\newcommand\Exp{{\mathsf{Exp}}}

\newcommand\E\varepsilon

\newcommand\emp{\underline\emptyset}

\newcommand\D{\mathscr{D}}
\newcommand\Rf{\mathscr{S}}
\newcommand\M{\mathscr{M}}
\newcommand\Pa{\mathscr{Q}}
\newcommand\N{\mathscr{N}}
\newcommand\Lf{\mathscr{L}}

\newcommand{\rules}[2]{\mbox{$\frac%
                     {\mbox{\small \rule[-5pt]{0pt}{14pt} $#1$}}
                     {\mbox{\small \rule[0pt]{0pt}{10pt}$#2$}}$}}
\mathcode`:="003A  
\mathcode`;="003B  
\mathcode`?="003F  
\mathcode`|="026A  
\mathcode`<="4268  
\mathcode`>="5269  

\mathchardef\ls="213C    
\mathchardef\gr="213E    

\newenvironment{todo}{\bigskip\hrule\medskip\noindent}{\medskip\hrule\bigskip}

\newtheorem{theorem}{Theorem}
\newtheorem{lemma}{Lemma}
\newdefinition{definition}{Definition}
\newproof{proof}{Proof}
\newtheorem{example}{Example}
\newtheorem{remark}{Remark}
\newtheorem{corollary}{Corollary}

\newcounter{counter1}
\newcounter{counter2}
\newcounter{counter3}
\begin{document}


\bibliographystyle{abbrv}

\begin{frontmatter}

\title{Automatic Equivalence Proofs for Non-deterministic Coalgebras}

\author[i1,i4]{Marcello~Bonsangue}
\ead{marcello@liacs.nl}

\author[i2]{Georgiana~Caltais}
\ead{gcaltais10@ru.is}

\author[i2]{Eugen-Ioan~Goriac}
\ead{egoriac10@ru.is}

\author[i3]{Dorel~Lucanu}
\ead{dlucanu@info.uaic.ro}

\author[i4,i5]{Jan~Rutten}
\ead{janr@cwi.nl}

\author[i5,i4,adr3]{Alexandra~Silva}
\ead{alexandra@cs.ru.nl}

\address[i1]{LIACS - Leiden University, The Netherlands}
\address[i2]{School of Computer Science - Reykjavik University, Iceland}
\address[i3]{Faculty of Computer Science - Alexandru Ioan Cuza University, Romania}
\address[i4]{Centrum Wiskunde \& Informatica, The Netherlands}
\address[i5]{Radboud University Nijmegen, The Netherlands}
\address[adr3]{HASLab / INESC TEC, Universidade do Minho, Braga, Portugal}

\begin{abstract}
A notion of generalized regular expressions for a large class of systems modeled as coalgebras, and an analogue of Kleene's theorem and Kleene algebra, were recently proposed by a subset of the authors of this paper. Examples of the systems covered include infinite streams, deterministic automata, Mealy machines and labelled transition systems. In this paper, we present a novel algorithm to decide whether two expressions are bisimilar or not. The procedure is implemented in the automatic theorem prover CIRC, by reducing coinduction to an entailment relation between an algebraic specification and an appropriate set of equations. We illustrate the generality of the tool with three examples: infinite streams of real numbers, Mealy machines and labelled transition systems.
\end{abstract}

\end{frontmatter}

\section{Introduction}
\label{sec:introd}

Regular expressions and \ans{C.2.1} finite deterministic automata (DFA's) constitute two of the most basic structures in computer science. Kleene's theorem~\cite{Kleene61} gives a fundamental correspondence between these two structures: each regular expression denotes a
language that can be recognized by a DFA and, conversely, the language accepted by a DFA
can be specified by a regular expression. Languages denoted by regular expressions are
called regular. Two regular expressions are (language) equivalent if they denote the
same regular language. Salomaa~\cite{salomaa} presented a sound and complete axiomatization
(later refined by Kozen in~\cite{kozen91,kozen-nerode}) for proving the equivalence of regular expressions.

The above programme was applied by Milner in~\cite{milner} to
process behaviours and labelled transition systems (LTS's). Milner
introduced a set of expressions for finite LTS's and proved an
analogue of Kleene's Theorem: each expression denotes the behaviour
of a finite LTS and, conversely, the behaviour of a finite LTS can
be specified by an expression (modulo bisimilarity).
Milner also provided an axiomatization for his expressions, with the
property that two expressions are provably equivalent if and only if
they are bisimilar.

Coalgebras arose in the last decade as a suitable mathematical framework to study state-based systems, such 
 as DFA's and LTS's. For a functor $\Gf \colon \mathbf{Set} \to \mathbf{Set}$, a $\Gf$-coalgebra or $\Gf$-system is a pair
$(S , g )$, consisting of a set $S$ of states and a function $g \colon S\to \Gf(S)$ defining the ``transitions''
of the states. We call the functor $\Gf$ the type of the system. For instance, DFA's can be readily
seen to correspond to coalgebras of the functor $\Gf(S) = 2 \times
S^A$ and image-finite LTS's are obtained by $\Gf(S) = \pow (S)^A$,
where $\pow$ is finite powerset.

For coalgebras of a large class of functors, a language of regular
expressions, a corresponding generalization of Kleene's theorem, and a
sound and complete axiomatization for the associated notion of
behavioral equivalence were  introduced in~\cite{brs_lmcs}.
Both the language of expressions and their axiomatization were derived, in a modular fashion, from the functor defining the type of the system.

Algebra and related tools can be successfully used for reasoning on properties of systems. In this paper, we present a novel method for checking for bisimilarity of generalized regular expressions using the coinductive theorem prover {\CIRC} \cite{goguen-lin-rosu-2000-ase,rosu-lucanu-2009-calco}. The main novelty of the method lies on the generality of the systems it can handle.
{\CIRC} is a metalanguage application implemented in Maude \cite{DBLP:conf/maude/2007}, and its target is to prove properties over infinite data structures. It has been successfully used for checking the equivalence of programs, and trace equivalence and strong bisimilarity of processes.
The tool may be tested online and downloaded from \url{https://fmse.info.uaic.ro/tools/Circ/}.

Determining whether two expressions are equivalent is important in order to be able to compare behavioral specifications.  
In the presence of a sound and complete axiomatization one can determine equivalence using algebraic reasoning. 
A coalgebraic perspective on regular expressions has however provided a more operational/algorithmic way of checking equivalence: one constructs a bisimulation relation containing both expressions. The advantage of the bisimulation approach is that it enables automation since the steps of the construction are fairly mechanic and require almost no ingenuity.
  
We remark that in theory it has been shown that both problems are in PSPACE~\cite{K08a,worthington}, but in practice bisimulation checking tends to be easier. We illustrate this with an example, to give the reader the feeling of the more algorithmic nature of bisimulation. We want to stress however that we are not underestimating the value of an algebraic treatment of regular expressions: on the contrary, as we will show later, the axiomatization plays an important role in guaranteeing termination of the bisimulation construction and is therefore crucial for the main result of this article.  

We show below a proof of the sliding rule: $a(ba)^* \equiv (ab)^*a$. The algebraic proof, using the rules and equations of Kleene algebra, needs to show the two containments
\[
{a(ba)^* \leq (ab)^*a} \qquad \text{ and } \qquad (ab)^*a \leq a(ba)^*
\]
and it requires some ingenuity in the choice of the equation applied in each step. We show the proof for the first inequality, the other would follow a similar proof pattern.
\[
\begin{array}{@{}lcl@{\qquad}l}
&&a(ba)^* \leq (ab)^*a \\
&\Leftarrow& a + (ab)^*a (ba) \leq (ab)^*a &\text{right-star rule}\\
&\iff & (1+ (ab)^*ab)a \leq (ab)^*a &\text{associativity and distributivity}\\
&\iff & (ab)^*a \leq (ab)^*a &\text{right expansion rule: $1+r^*r = r^*$}
\end{array}
\]

For the coalgebraic proof, we build incrementally, and rather mechanically, a bisimulation relation containing the pair  $(a(ba)^*, (ab)^*a)$.  We start with the pair we want to prove equivalent and then we close the relation  with respect to syntactic language derivatives, also known as  {\em  Brzozowski derivatives}. In the current example, the bisimulation relation would contain three pairs: 
\[
R = \{ (a(ba)^*, (ab)^*a), ((ba)^*, b(ab)^*a+1), (0,0)\}
\]
where $1$ and $0$ are, respectively, the regular expressions denoting the empty word and the empty language. In constructing this relation, no decisions were made, and hence the suitability of bisimulation construction as an automatic technique to prove equivalence of regular expressions. 

The main contributions of this paper can be summarized as follows. We present a decision procedure to determine equivalence of generalized regular expressions, which specify behaviours of many types of transition systems, including Mealy machines, labelled transition systems and infinite streams. The valid expressions for each system are type-checked automatically in the tool. We illustrate the decision procedure we devised by applying it to several examples. As a vehicle of implementation, we choose \CIRC, a coinductive theorem prover which has already been explored for the construction of bisimulations. To ease the implementation in \CIRC,  we present the algebraic specifications' counterpart of the coalgebraic framework of the generalized regular expressions mentioned above.
This enables us to automatically derive algebraic specifications that model the language of expressions, and to define an appropriate equational entailment relation which mimics our decision procedure for checking behavioural equivalence of expressions. The implementation of both the algebraic specification and the entailment relation in {\CIRC} allows for automatic reasoning on the equivalence of expressions.

The present paper is an extended version of the conference paper~\cite{sbmf}. In comparison with the aforementioned paper we have extended the tool to deal with non-deterministic systems. More precisely, we have included the powerset function in the class of functors considered. Moreover, we have included all the proofs, more examples and additional explanations on the theory behind and implementation of the tool.

\vspace{1ex}
\noindent
\textit{Organization of the paper.} Section~\ref{sec:prelim} recalls the basic definitions of the language associated to a non-deterministic functor. Section~\ref{sec:dp} describes the decision procedure to check equivalence of regular expressions. Section~\ref{sec:algSpec} formulates the aforementioned language as an algebraic specification, which paves the way to implement in {\CIRC} the procedure to decide equivalence of expressions. The implementation of the decision procedure and its soundness are described in Section~\ref{sec:dec-proced}.
In Section~\ref{sec:caseStudy} we show, by means of several examples,
how one can check bisimilarity, using {\CIRC}.
Section~\ref{sec:concl} contains concluding remarks and pointers for future work.

\section{Regular Expressions for Non-deterministic Coalgebras}
\label{sec:prelim}

In this section, we briefly recall the basic definitions in
\cite{brs_lmcs}.

Let \textbf{Set} denote the category of sets (represented by capital letters $X, Y, \ldots$)
and functions (represented by lower case letters $f, g, \ldots$).
We write $Y^X$ for the family of functions from $X$ to $Y$ and $\pow(X)$ for the collection of finite subsets of a set $X$.
The product of two sets $X, Y$ is written as $X \times Y$
and has the projections functions $\pi_1$ and $\pi_2$:
$X
\mathrel{{\overset{\pi_1}{\longleftarrow}}}
X\times Y
\mathrel{{\overset{\pi_2}{\longrightarrow}}}
Y$.
We define
\(
X\myplus Y = X\uplus Y \uplus \{\bot,\top\}
\)
where $\uplus$ is the disjoint
union of sets, with injections
$
X \mathrel{{\overset{\kappa_1}{\longrightarrow}}} X
\uplus
Y \mathrel{{\overset{\kappa_2}{\longleftarrow}}} Y$.
Note that the set $X\myplus Y$ is different from the classical
coproduct of $X$ and $Y$ (which we shall denote by $X+Y$), because of the
two extra elements $\bot$ and
$\top$.
These extra elements are used to represent, respectively,
underspecification and inconsistency in the specification of some
systems.

For each of the operations defined above on sets, there are analogous
ones on functions. Let $f\colon X\to Y$, $f_1\colon X\to Y$ and $f_2\colon Z\to W$. We
define the following operations:
\[
\begin{array}{l@{\hspace{1.0cm}}l}
f_1\times f_2 \colon X\times Z \to Y\times W & f_1\myplus f_2 \colon X\myplus
Z \to Y\myplus W\\[.6ex]
(f_1\times f_2)(x,z) = <f_1(x),f_2(z)> &
(f_1\myplus f_2)(c) = c,\ c\in\{\bot,\top\}\\[.6ex]
 &
(f_1\myplus f_2)(\kappa_i(x)) = \kappa_i(f_i(x)),\ i\in\overline{1,2}\\[.6ex]
f^A \colon X^A \to Y^A  & \pow(f) \colon \pow(X) \to \pow(Y)\\
f^A (g) = f\circ g & \pow(f)(X_1) = \{y \in Y \mid f(x) = y, x\in X_1\}
\end{array}
\]

\begin{remark}
\ans{C.1.16}
For the sake of brevity, we use the notation $i \in \overline{1,n}$ as a shorthand for $i \in \{1, \ldots, n\}$.
\end{remark}

Note that in the definition above we are using the same symbols that we defined above for
the
operations on sets. It will always be clear from the context which
operation is being used.

In our definition of non-deterministic functors we will use constant
sets
equipped with an information order. In particular, we will use
join-semilattices. A (bounded) join-semilattice is a set $\Bl$
equipped
with a binary operation $\vee_\Bl$ and a constant $\bot_\Bl \in \Bl$, such
that $\vee_\Bl$ is commutative, associative and idempotent. The
element $\bot_\Bl$ is neutral with respect to $\vee_\Bl$. As usual, $\vee_\Bl$
gives rise to a partial ordering $\leq_\Bl$ on the elements of $\Bl$: $
b_1 \leq_\Bl b_2 \Leftrightarrow b_1\vee_\Bl b_2 = b_2$.
Every set $S$
can be mapped into a join-semilattice by taking $\Bl$ to be
the
set of all finite subsets of $S$ with empty set as $\bot_{\Bl}$, and union as join.

\paragraph{\bf Coalgebras} A coalgebra is a pair $(S, g\colon
S\to \Gf(S))$, where $S$
is a set of states and $\Gf\colon \mathbf{Set}\to \mathbf{Set}$ is a functor.
The functor $\Gf$, together with the
function $g$, determines the {\em transition structure} (or
dynamics) of the $\Gf$-coalgebra~\cite{Rutten00}.

\ans{C.1.2} A coalgebra $(S,g)$ is \emph{finite} if $S$ is a finite set.

\begin{definition}[Bisimulation]\label{def:bis}
Let $(S, f)$ and $(T,g)$ be two $\Gf$-coalgebras. We call a relation
$R \subseteq S\times T$ a {\em bisimulation\/}~\cite{DBLP:journals/iandc/HermidaJ98}
iff
\[
(s,t)\in R \Rightarrow (f(s), g(t))\in \overline \Gf(R)
\]
where $\overline \Gf(R)$ is defined as
$
\overline \Gf(R) = \{ (\Gf(\pi_1)(x),\Gf(\pi_2)(x)) \mid x \in \Gf(R) \}
$.
\end{definition}

\ans{C.1.3} We write $s\sim_\Gf t$ whenever there exists a bisimulation relation
containing $(s,t)$ and we call $\sim_\Gf$ the bisimilarity relation. It is of interest to remark that the relation $\sim_\Gf$ is an equivalence relation.
We shall drop the subscript $\Gf$ whenever the functor $\Gf$ is clear
from the context. In the literature, one finds different definitions of bisimulation or behavioral equivalence~\cite{staton}. For the class of functors we consider here the different notions coincide and therefore we will not discuss them. 

\paragraph{\bf Non-deterministic functors} They are functors $\Gf \colon {\bf{Set}} \rightarrow {\bf{Set}}$
built inductively from the identity,
and constants, using $\times$,
$\myplus$, $(-)^A$ and $\pow$:

\begin{equation}\label{eq:fun-gram}
\pf\ni \Gf \,::\!=\, \id \mid \Bl \mid \Gf\myplus \Gf \mid \Gf\times
\Gf \mid \Gf^A \mid \pow \Gf
\end{equation}
where $\Bl$ is a finite join-semilattice and $A$ is a finite set.
Typical examples of non-deterministic functors include $\Rf=\Bl\times \id$, $\M =
(\Bl\times \id)^A$,
$\D = 2 \times \id^A$, $\Pa = (1 \myplus \id)^A$,  $\N = 2 \times \pow(\id)^A$ and $\Lf = 1 \myplus \pow(\id)^A$. These
functors represent, respectively, the type of streams, Mealy,
deterministic, partial deterministic automata, non-deterministic automata and labeled transition systems with explicit termination. $\Rf$-bisimulation is stream equality, whereas $\D$-bisimulation coincides with language equivalence.

\begin{remark}
\ans{C.3.7}
As stated in~\cite{brs_lmcs}, the use of join-semilattices for constant functors and the sum $\myplus$ instead of the ordinary product enabled the use of underspecification and inconsistency ({\it i.e.}, $\top$ and $\bottom$, respectively) in the specification of systems, and moreover, has allowed the whole framework to be studied in the category {\bf Set}. Even though underspecification and inconsistency can be captured by a semilattice structure, and the axiomatization provides the set of expressions with a join-semilattice structure (therefore allowing the work directly in the category of join-semilattices), remaining in the category {\bf Set} was chosen for simplicity.
\end{remark}

Next, we give the definition of  the ingredient relation, which
relates a non-deterministic functor $\Gf$ with its {\em ingredients}, {\em
i.e.}, the functors used in its inductive construction. We shall use
this relation later for typing our expressions.

\begin{definition}\label{def:ingred}

Let $\lhd\subseteq \pf\times \pf$ be the least reflexive and transitive relation on
non-deterministic functors
such that
\[
\Gf_1\triangleleft \Gf_1\times \Gf_2,\ \ \ 
\Gf_2\lhd \Gf_1\times \Gf_2,\ \ \ 
\Gf_1\lhd \Gf_1\myplus \Gf_2,\ \ \ 
\Gf_2\lhd \Gf_1\myplus \Gf_2,\ \ \ 
\Gf\lhd \Gf^A, \ \ \ \Gf \lhd \pow\Gf.
\]
\end{definition}
Here and throughout this document we use $\F \lhd \Gf$ as a shorthand
for $(\F,\Gf)\in \lhd$. If $\F \lhd \Gf$, then $\F$ is said to  be an
\emph{ingredient} of $\Gf$. For example, $2$, $\id$, $\id^A$ and $\D$
itself are all the ingredients of the deterministic automata functor
$\D$.

\paragraph{\bf  A language of regular expressions for non-deterministic coalgebras} We now associate a language of expressions $\Exp_\Gf$ with each non-deterministic functor $\Gf$.

\begin{definition}[Expressions]\label{def:expr}
Let $A$ be a finite set, $\Bl$ a finite join-semilattice and $X$ a set
of fixed-point variables. The set $\Exp$ of all {\em expressions\/} is given
by the following grammar, where $a\in A$, $b\in \Bl$ and $x\in X$:
\begin{equation}\label{eq:grammar}
\begin{array}{lcl}
\E &::\!=&  x \mid \E \oplus \E \mid \gamma\\
\end{array}
\end{equation}
where $\gamma$ is a {\em guarded expression} given by:
\begin{equation}\label{eq:grammar2}
\begin{array}{lcl}
\gamma &::\!=& \emp \mid \gamma \oplus \gamma \mid \mu
x.\gamma
   \mid b \mid l<\E> \mid r<\E> \mid l[\E] \mid r[\E] \mid a(\E) \mid \{\E\} \\
\end{array}
\end{equation}
\end{definition}
In the expression $\mu x.\gamma$, $\mu$ is a binder for
all the free occurrences of $x$ in $\gamma$. Variables that are not bound are
free. A {\em closed expression} is an expression
without free occurrences of fixed-point variables $x$. We denote the set of closed
expressions by $\Exp^c$.

The language of expressions for non-deterministic coalgebras
is a generalization of the classical notion of regular expressions:
$\emp$, $\eps_1 \oplus \eps_2$ and $\mu x.\gamma$
play similar roles to the regular expressions denoting empty language, the union of languages and the
Kleene star.
\ans{C.1.4} Moreover, note that, not unexpectedly, in~\cite{brs_lmcs}, $\oplus$ was axiomatized as an associative, commutative and idempotent operator, with $\emp$ as a neutral element. 
The expressions $l\langle\eps\rangle$, $r\langle\eps\rangle$,
$l[\eps]$, $r[\eps]$, $a(\eps)$ and $\{\E \}$ specify the left and right hand-side
of products and sums, function application and singleton sets, respectively.
Next, we present a type assignment system for associating
expressions to non-deterministic functors. This will allow us to associate with each functor
$\Gf$ the expressions $\E\in \Exp^c$ that are valid specifications of
$\Gf$-coalgebras.

\begin{definition}[Type system]\label{def:ts}
We now define a typing relation $\mathbin\vdash\subseteq \Exp \times \pf
\times \pf $ that will associate an expression $\E$ with two non-deterministic functors $\F$ and $\Gf$, which are related by the
ingredient relation ($\F$ is an ingredient of $\Gf$). We shall write
$\vdash \E\colon \F\lhd \Gf$ for $(\E,\F,\Gf) \in \;\vdash$.  The rules that define $\vdash$ are the following:
\[
\renewcommand{\arraystretch}{0.5}
\begin{array}{@{}ccc@{}}
\rules{}{\vdash \emp \colon \F\lhd \Gf }&
\rules{}{\vdash b\colon \Bl\lhd \Gf}\,\, (b \in \Bl)&
\rules{}{\vdash x \colon \Gf\lhd \Gf }\,\, (x \in X)\\\\
\rules{\vdash\E\colon \Gf\lhd \Gf}
     {\vdash \mu x.\E \colon \Gf\lhd \Gf}&
\rules{\vdash \E_1 \colon \F\lhd \Gf\;\;\;\; \vdash\E_2\colon \F\lhd \Gf}{\vdash \E_1\oplus\E_2 \colon \F\lhd \Gf} &
\rules{\vdash \E \colon \Gf\lhd \Gf}
     {\vdash \E \colon \id\lhd \Gf}\\\\
     \rules{\vdash \E\colon \F_2\lhd \Gf}
     {\vdash r[\E] \colon \F_1\myplus \F_2\lhd \Gf}
     &
\rules{\vdash \E\colon \F\lhd \Gf}
     {\vdash a(\E) \colon \F^A\lhd \Gf}\,\, (a \in A)
&
\rules{\vdash \E\colon \F_1\lhd \Gf}
     {\vdash l<\E> \colon \F_1\times \F_2\lhd \Gf}\\\\
\rules{\vdash \E\colon \F_2\lhd \Gf}
     {\vdash r<\E> \colon \F_1\times \F_2\lhd \Gf}&
\rules{\vdash \E\colon \F_1\lhd \Gf}
     {\vdash l[\E] \colon \F_1\myplus \F_2\lhd \Gf} & \rules{\vdash \E\colon \F_1\lhd \Gf}{\vdash \{\E\} \colon \pow \F_1 \triangleleft \Gf
    }
\end{array}
\]
\end{definition}
We can now formally define the set of $\Gf$-expressions: well-typed
expressions associated with a non-deterministic functor $\Gf$.

\begin{definition}[$\Gf$-expressions]\label{def:g-expr}
Let $\Gf$ be a non-deterministic functor and $\F$ an ingredient of $\Gf$.
We define $\Exp_{\F\lhd \Gf}$ by:
\[
\Exp_{\F\lhd \Gf} = \{\E \in \Exp^c \mid\ \vdash\E \colon \F\lhd \Gf\}\,.
\]
We define the set $\Exp_\Gf$ of well-typed {\em
$\Gf$-expressions\/} by $\Exp_{\Gf\lhd \Gf}$.
\end{definition}

In \cite{brs_lmcs}, it was proved that
the set of $\Gf$-expressions for a given non-deterministic functor $\Gf$
has a coalgebraic structure:
\[
\delta_{\Gf} \colon {\Exp}_{\Gf} \to {\Gf}({\Exp}_{\Gf})
\]
More precisely, in \cite{brs_lmcs}, which we
refer to for the complete definition of $\delta_{\Gf}$, the authors defined a function
$
\delta_{\F \lhd \Gf} \colon {\Exp}_{\F\lhd \Gf} \to {\F}({\Exp}_{\Gf})
$
and then set $\delta_\Gf = \delta_{\Gf\lhd \Gf}$.

The coalgebraic structure on the set of expressions enabled the proof of a Kleene like theorem.
\begin{theorem}[Kleene's theorem for non-deterministic coalgebras]\label{thm:kleene}
Let $\Gf$ be a non-deterministic functor.
\begin{enumerate}
\item For any $\E\in \Exp_\Gf$, there exists a finite $\Gf$-coalgebra $(S,g)$  and $s\in S$ such that $\E\sim s$.
\item For every \ans{C.1.6, C.2.3} finite $\Gf$-coalgebra $(S,g)$  and $s\in S$ there exists an expression $\E_s\in \Exp_\Gf$ such that $\E_s\sim s$.
\end{enumerate}
\end{theorem}

In order to provide the reader with intuition over the notions presented above, we illustrate them with an example.
\begin{example}\label{ex:streams1}
Let us instantiate the definition of $\Gf$-expressions to the functor of
streams \ans{C.1.13} $\Rf = \Bl\times \id$ (the ingredients of this functor are $\Bl$, $\id$ and $\Rf$ itself).
Let $X$ be a set of
(recursion or) fixed-point variables. The set $\Exp_\Rf$ of {\em stream
expressions\/} is given by the set of closed, guarded expressions generated by the following BNF grammar. For $x \in X$: \ans{C.1.7}
\begin{equation}\label{eq:grammarStr}
\begin{array} {l@{\;}l@{\;}c@{\;}l}
\Exp_\Rf\ni &\E &::\!=&  \emp \mid \E \oplus \E \mid
        \mu x.\E \mid
x \mid l<\tau> \mid r<\E> \\
&\tau &::\!=& \emp \mid b \mid \tau\oplus\tau\\
\end{array}
\end{equation}
\end{example}
Intuitively, the expression $l\langle b \rangle$
is used to specify that the head of the stream is $b$, while
$r\langle \E \rangle$ specifies a stream whose tail behaves
as specified by $\E$.
For the two element join-semilattice $\Bl=\{0,1\}$ (with $\bot_\Bl=0$) examples of well-typed expressions include
$\emp$, $l<1>\oplus r<l<\emp>>$ and $\mu x. r<x> \oplus
l<1>$. The expressions $l[1]$, $l<1>\oplus 1$ and $\mu x. 1$ are examples of
non well-typed
expressions for $\Rf$, because the functor $\Rf$ does not involve $\myplus$, the
subexpressions in the sum have different type, and
recursion is not at the outermost level ($1$ has type $\Bl\lhd \Rf$),
respectively.

By applying the definition in \cite{brs_lmcs}, the coalgebra structure on expressions $\delta_\Rf$ would be given by:
\[
\begin{array}{lcl}
\multicolumn{3}{l}{\delta_\Rf \colon \Exp_\Rf \to \Bl\times \Exp_\Rf}\\
\delta_\Rf(\emp) &=& <\bot_{\Bl}, \emp>\\
\delta_\Rf(\E_1 \oplus \E_2) &=& <b_1\vee b_2, \E_1' \oplus\E_2'>\text{ where } <b_i, \E_i'> = \delta_\Rf(\E_i), \ i\in\overline{1,2}\\
\delta_\Rf(\mu x.\E) &=&  \delta_\Rf(\E[\mu x.\E/x]) \\
\delta_\Rf(l<\tau> ) &=& <\delta_{\Bl\lhd \Rf}(\tau),\emp>\\
\delta_\Rf(r<\E> ) &=&  <\bot_\Bl,\E>\\
\delta_{\Bl\lhd \Rf}(\emp) &=& \bot_B\\
\delta_{\Bl\lhd \Rf}(b) &=& b\\
 \delta_{\Bl\lhd \Rf}(\tau\oplus\tau') &=&  \delta_{\Bl\lhd \Rf}(\tau)  \vee \delta_{\Bl\lhd \Rf}(\tau')\\
\end{array}
\]
The proof of Kleene's theorem provides algorithms to go from expressions to streams and
vice-versa. We illustrate it by means of examples.

Consider the following stream:
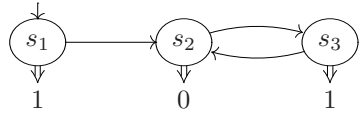
\[
\xymatrix@C=1.2cm@R=.25cm{
& & &\\
*++[o][F]{s_1} \ar@{<-}[u]\ar[r]\ar@{=>}[d] & *++[o][F]{s_2}\ar@/^/[r]\ar@{=>}[d] & *++[o][F]{s_3}\ar@/^/[l]\ar@{=>}[d]\\
1&0&1
}
\]
We draw the stream with an automata-like flavor. The transitions indicate the tail of the stream represented by a state and the output value the head. In a more traditional notation, the above automata represents the infinite stream $(1,0,1,0,1,0,1,\ldots)$. 

\ans{C.2.4} To compute expressions $\E_1$, $\E_2$ and $\E_3$ equivalent to $s_1$, $s_2$ and $s_3$ we associate with each state $s_i$ a variable $x_i$ and get the equations:
\[
\eps_1 = \mu x_1. l<1> \oplus r<x_2>\ \ \ \eps_2 = \mu x_2 .l<0> \oplus r<x_3>\ \ \
\eps_3 = \mu x_3. l<1> \oplus r<x_2>
\]
As our goal is to remove all the occurrences of free variables in our expressions, we proceed as follows.
First we substitute $x_2$ by $\eps_2$ in $\eps_1$, and $x_3$ by $\eps_3$ in $\eps_2$, and obtain the following expressions:
\[
\eps_1 = \mu x_1. l<1> \oplus r<\E_2>\ \ \
\eps_2 = \mu x_2 .l<0> \oplus r<\E_3>\ \ \
\]
Note that at this point $\eps_1$ and $\eps_2$ already denote closed expressions. Therefore, as a last step, we replace $x_2$ in $\eps_3$ by $\eps_2$ and get the following closed expressions:
\[
\eps_1 = \mu x_1. l<1> \oplus r<\E_2>\ \ \
\eps_2 = \mu x_2 .l<0> \oplus r<\E_3>\ \ \
\eps_3 = \mu x_3. l<1> \oplus r< \mu x_2 .l<0> \oplus r<x_3>>
\]
satisfying, by construction, $\E_1\sim s_1$,  $\E_2\sim s_2$ and  $\E_3\sim s_3$.

For the converse construction, consider the expression $\E = (\mu x. r<x>) \oplus l<1>$. We construct an automaton by repeatedly applying the coalgebra structure on expressions $\delta_\Rf$, modulo associativity, commutativity and idempotence (ACI) of $\oplus$ in order to guarantee finiteness.

\ans{C.1.8} First, note that $\delta_\Rf(\mu x . r <x>) = \delta_\Rf(r<\mu x.r<x>>) = <\bottom_\Bl, \mu x . r <x>>$.
Applying the definition of $\delta_\Rf$ above, we have:
\[
\delta_\Rf(\E) = <1,  (\mu x. r<x>) \oplus \emp> \text{ and } \delta_\Rf((\mu x. r<x>) \oplus \emp) = <0, (\mu x. r<x>) \oplus \emp>
\]
which leads to the following stream (automaton):
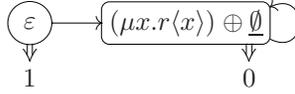
\[
\xymatrix@C=1.5cm@R0.25cm{
*++[o][F]{\E}\ar[r]\ar@{=>}[d]  &
*+[l][F-:<3pt>]{(\mu x. r<x>) \oplus \emp}\ar@(dr,ur)\ar@{=>}[d] \\
1&0
}
\]

At this point, we want to remark that the direct application of $\delta_\Rf$, without ACI, might generate infinite automata. Take, for instance, the expression $\E = \mu x. r<x\oplus x>$ . Note that  $\delta_\Rf(\mu x . r <x\oplus x>) = <0,\E\oplus \E>$, $\delta_\Rf(\E\oplus \E) = <0,(\E\oplus \E)\oplus (\E\oplus \E)>$, and so on. This would generate the infinite automaton
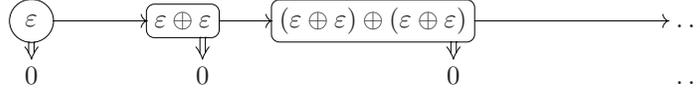
\[
\xymatrix@C=1.5cm@R0.25cm{
*++[o][F]{\E}\ar[r]\ar@{=>}[d]  &
*+[l][F-:<3pt>]{\E\oplus\E}\ar@{=>}[d] \ar[r] & *+[l][F-:<3pt>]{(\E\oplus \E)\oplus (\E\oplus \E)}\ar@{=>}[d] \ar[r] & \ldots \\
0&0&0& \ldots
}
\]
\noindent instead of the intended, simple and very finite, automaton
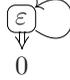
\[
\xymatrix@C=1.5cm@R0.25cm{
*+[l][F-:<3pt>]{\E}\ar@(dr,ur)\ar@{=>}[d] \\
0
}
\]
In order to guarantee finiteness, one needs to identify the expressions modulo associativity, commutativity and idempotence (ACI), as we will discuss further in this paper. Moreover, the axiom $\E\oplus \emp \equiv \emp$ could also be used in order to obtain smaller automata, but it is not crucial for termination. 

Throughout the paper, we will often use streams as a
basic example to illustrate the definitions. It should be remarked
that the framework is general enough to include more complex examples,
such as deterministic automata, automata on guarded strings, Mealy machines and labelled transition systems.
The latter two will be used as examples in Section~\ref{sec:caseStudy}.

\section{A Decision Procedure for the Equivalence of Generalized Regular Expressions}\label{sec:dp}

In this section, we briefly describe the decision procedure to determine whether two expressions are equivalent or not. 

The key observation is that point $1.$ of Theorem~\ref{thm:kleene}  above guarantees that each expression in the language for a given system can always be associated to a \emph{finite} coalgebra. Given two expressions $\E_1$ and $\E_2$ in the language $\Exp_\Gf$ of a given functor $\Gf$ we can decide whether they are equivalent by constructing a \emph{finite} bisimulation between them. This is because the finite coalgebra generated from an expression contains precisely all states that one needs to construct the equivalence relation. Even though this might seem like a trivial observation, it has very concrete consequences: for (all well-typed) generalized regular expressions we can always either determine that they are bisimilar, and exhibit a proof in the form of a bisimulation, or conclude that they are not bisimilar and pinpoint the difference by showing why the bisimulation construction failed. Hence, we have a decision procedure for equivalence of generalized regular expressions.

We will give the reader a brief example on how the equivalence check works. Further examples, for different types of systems, including examples of non-equivalence, will appear in Section~\ref{sec:caseStudy}. 

We will show that the stream expressions $\eps_1 = \mu x . r<x> \oplus l<0> $ and $\eps_2 =  r< \mu x . r<x> \oplus l<0> > \oplus l<0> $ are equivalent. In order to do that, we have to build a bisimulation relation  $\R$ on expressions for the stream functor $\Rf$, defined above, such that  $(\eps_1, \eps_2) \in \R$. We do this in the following way: we start by taking $\R=\{(\E_1,\E_2)\}$ and we check whether this is already a bisimulation, by applying  $\delta_\Rf$ to each of the expressions and checking whether the expressions have the same output value and, moreover,  that no new pairs of expressions (modulo associativity, commutativity and idempotence, for more details see page~\pageref{just:ACI}) appear when taking transitions. If new pairs of expressions appear we add them to $\R$ and repeat the process. Intuitively, for this particular example, the transition structure can be depicted as follows:

\begin{figure}[h]
\centering
$\xymatrix@C=1.2cm@R=0.8cm{
& {\eps_1}\ar[d]\ar@{..}[rr]^ \R & & {\eps_2}\ar[d] & \R = \{(\eps_1, \eps_2)\} \\
&{\eps_1 }  \ar[d]& &   {\eps_1}\ar[d]\ar@{..}_{\text{not yet in }\R \text{; add it}} [ll]& \R = \{(\eps_1, \eps_2), (\eps_1, \eps_1)\}\\
& {\eps_1}\ar@{..}[rr]^ \R & & {\eps_1}& \checkmark
}$
\caption{Bisimulation construction}
\label{fig:streams-BC}
\end{figure}
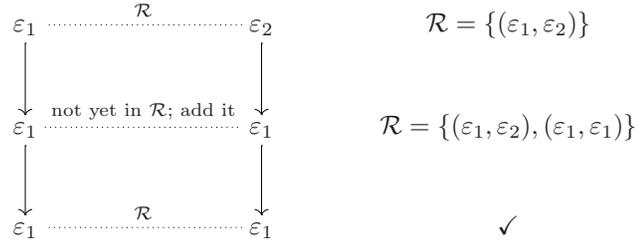

Here, we omit the output values of the expressions, which are all $0$. 
In the figure above, we use the notation $\xymatrix@C=0.7cm@R=0.7cm{
 {\eps_1}\ar@{-}[r]^\R&\eps_2}$ to denote $(\eps_1,\eps_2)\in\R$.
As illustrated in Figure~\ref{fig:streams-BC}, $\R = \{(\eps_1, \eps_2), (\eps_2, \eps_2)\}$ is closed under transitions and is therefore a bisimulation. Hence, $\eps_1$ and $\eps_2$ are bisimilar and specify the same infinite stream (concretely, the stream with only zeros).

\section{An Algebraic View on the Coalgebra of Generalized Regular\\ Expressions}
\label{sec:algSpec}

Recall that our goal is to reason about equality of generalized regular expressions in a fully automated manner. As we showed in the introduction, obtaining this equality can be achieved in two distinct ways: either algebraically, reasoning with the axioms, or coalgebraically, by constructing a bisimulation relation. The latter, because of its algorithmic nature, is particularly suited for automation. Automatic constructions of bisimulations have been widely explored in {\CIRC} and we will use this tool to implement our algorithm. This section contains material that enables us to soundly use {\CIRC}. We want to stress however that the main result of the paper is the description of a {\em decision procedure} to determine whether two expressions are equivalent or not. This procedure in turn could be implemented in any other suitable tool or even as a standalone application. Choosing {\CIRC} was natural for us, given the pre-existent work on bisimulation constructions. In Section~\ref{sec:dec-proced}, we show that the process of generating the $\Gf$-coalgebras associated to expressions by repeatedly applying $\delta_{\Gf}$ and normalizing the expressions obtained at each step is closely related to the proving mechanism already existent in {\CIRC}.

In Section~\ref{sec:prelim}, we have introduced a (theoretical) framework which, given a functor $\itG$, allows
for the uniform derivation of 1) a language $\itExp_\itG$ for specifying behaviors of
$\itG$-systems, and 2) a coalgebraic structure on $\itExp_\itG$, which provides an
operational semantics to the set of expressions.
In this context, given that {\CIRC} is based on algebraic specifications, we need two things in order to reach our final goal:
\begin{itemize}\itemsep2pt
\item extend and adapt the framework of Section~\ref{sec:prelim} in order to enable the implementation of a tool which allows the automatic derivation of \emph{algebraic specifications} that model 1) and 2) above, to deliver to {\CIRC};
\item provide a decision procedure, implemented in {\CIRC} based on an \emph{equational entailment relation}, in order to check
bisimilarity of expressions.
\end{itemize}

In the rest of the paper we will present the algebraic setting for reasoning on bisimilarity of generalized regular expressions.
\ans{C1.12} A brief overview on the parallel between the coalgebraic concepts in~\cite{brs_lmcs} and their algebraic correspondents introduced in this section is provided later, in Figure~\ref{fig:CoVsAlg}.

\begin{paragraph}{\bf{Algebraic specifications}}
An {\em algebraic specification} is a triple
${\mathcal E} = (S,\Sigma,E)$, where $S$ is a set of {\em sorts}, $\Sigma$ is a
{\em $S$-sorted signature} and $E$ is a set of
{\em conditional equations} of the form
$(\forall X)\,t=t'{\tt ~if~}(\bigwedge_{i\in I} u_i=v_i)$, where
$t$, $t'$, $u_i$, and $v_i$ ($i\in I$ -- a set of indices for the conditions) are $\Sigma$-terms with variables in $X$.
We say that the \textit{sort of the equation} is $s$
whenever $t, t' \in {\mathcal T}_{\Sigma, s}({X})$.
Here, ${\mathcal T}_{\Sigma, s}({X})$ denotes the set of terms of sort $s$ of the
$\Sigma$-algebra freely generated by $X$.
If $I=\{\}$ then the equation is \emph{unconditional} and may
be written as $(\forall X)\,t=t'$.

Let $\vdash$ be the
{\em equational entailment (deduction) relation} defined as in \cite{Goguen92order-sortedalgebra}.
For consistency reasons, we write $\mathcal E \vdash e$ whenever
equation $e$ is deducible from the equations $E$ in $\mathcal E$ by reflexivity, symmetry, transitivity, congruence or substitutivity ({\it{i.e.}}, whenever $E \vdash e$).

\end{paragraph}

\medskip

In this paper, the algebraic specifications of coalgebras of generalized regular expressions are built on top of definitions based on grammars in Backus-Naur form (BNF) such as (\ref{eq:fun-gram}) and (\ref{eq:grammar}). Therefore, in what follows, we introduce the general technique for transforming BNF notations into algebraic specifications.

\medskip
\begin{paragraph}{\bf From BNF grammars to algebraic specifications}
The general rule used for translating definitions based on BNF grammars into algebraic specifications is as follows: each syntactical category and vocabulary is considered as a sort and each production is considered as a constructor operation or a subsort relation.

For instance, according to the grammar (\ref{eq:fun-gram}) of non-deterministic functors, we have a sort $\SltN$ -- representing the vocabulary of join-semilattices $\itB$, a sort $\AlphbN$ -- for the vocabulary of the alphabets $A$, a sort $\Fun$ -- associated to the syntactical category of the non-deterministic functors $\Gf$, a subsort relation
$\SltN \ls \Fun$ representing the production $\Gf \,::\!\!= \itB$, and constructor operations for the other productions.

\ans{C.1.14} Generally, each production $A \,::\!\!= {\it{rhs}}$ gives rise to a constructor $({\it{rhs}}) \rightarrow (A)$, the direction of the arrow being reversed.
For instance, for grammar (\ref{eq:fun-gram}), the production $\Gf \,::\!\!= \id$ is represented by a constant (nullary operation) $\id : {}\to \it \Fun$, and
the sum construction by the binary operation $\_\myplus\!\_ \,\,:\, \it \Fun~\Fun \to \Fun$.

\begin{remark}
\ans{C.3.16}
Note that the above mechanism for translating BNF grammars into algebraic specifications makes use of subsort relations for representing productions such as $\Gf\,::\!\!=\Bl$. This is because {\CIRC} works with order-sorted algebras, and we want to keep the algebraic specifications of non-deterministic functors as close as possible to their implementation in {\CIRC}. However, order-sorted algebras can be reduced to many-sorted algebras~\cite{Goguen92order-sortedalgebra}, where a subsort relation $s \ls s'$ is modeled by an inclusion operation $c_{s,s'}\,:\, s \rightarrow s'$. This way, even if we use order-sorted algebras, we remain in the framework of circular coinduction.
\end{remark}

\end{paragraph}
 
\medskip
The algebraic specifications of coalgebras of generalized regular expressions are defined in a modular fashion, based on the specifications of:
\begin{itemize}\itemsep0pt
\item non-deterministic functors ($\Gf$);
\item generalized regular expressions ($\eps \in \Exp_\Gf$);
\item ``transition" functions ($\delta_\Gf$);
\item ``structured" expressions ($\sigma \in \F(\Exp_\Gf)$, for all $\F$ ingredients of $\Gf$).
\end{itemize}

Moreover, recall that for a non-deterministic functor $\Gf$, bisimilarity of $\Gf$-expressions is decided based on the relation lifting $\overline{\Gf}$ over ``structured" expressions in $\Gf(\Exp_\Gf)$ (Definition~\ref{def:bis}). Therefore, the deduction relation $\vdash$ has to be extended to allow a restricted contextual reasoning over ``structured" expressions in $\F(\Exp_\Gf)$, for all ingredients $\F$ of $\Gf$.

The aforementioned algebraic specifications and the extension of $\vdash$ are modeled as follows.

\medskip
\begin{paragraph}{\bf The algebraic specification
of a non-deterministic functor $\Gf$} It includes:
\begin{itemize}
\item the translation of the BNF grammar~(\ref{eq:fun-gram}), as presented above;
\item the specification of the functor ingredients, given by a sort $\Ingred$ and a constructor $\_\triangleleft\_{} \ct \Fun~\Fun \to \Ingred$ (according to Definition~\ref{def:ingred});
\item the specification of each alphabet $A=\{a_1,\ldots,a_n\}$ occurring in the definition of $\Gf$: this consists of a subsort $A \ls \Alphb$, a constant $a_i:{}\to A$ for $i\in\overline{1,n}$, and a distinguished constant $A$ of sort $\AlphbN$ used to refer the alphabet in the definition of the functor;
\item the specification of each semilattice  $\itB=(\{b_1,\ldots,b_n\},\lor,\bot_B)$ occurring in the definition of $\Gf$: this consists of a subsort $\itB \ls \Slt$, a constant $b_i:{}\to \itB$ for $i\in\overline{1,n}$, a distinguished constant $\itB$ of sort $\SltN$ used to refer the corresponding semilattice in the definition of the functor, and the equations defining $\lor$ and $\bottom_B$ (this should be one of $b_i$);
\item an equation defining $\Gf$ (as a functor expression).
\end{itemize}
\end{paragraph}

\medskip
\begin{paragraph}{\bf The algebraic specification of generalized regular expressions} It consists of:
\begin{itemize}
\item (according to the BNF grammar
in Definition~\ref{def:expr}) a sort $\Exp$ representing expressions $\eps$, $\Fixpv$ the sort for the vocabulary of the fixed-point variables, and $\Slt$ the sort for the elements of semilattices. Moreover, we consider constructor operations for all the productions. For example, the production $\eps\, ::\!\!= \eps \oplus \eps$ is represented by an operation $\_\oplus\_\,\ct\Exp\,\,\Exp \to \Exp$, \ans{C.4.3} and $\eps\, ::\!\!= \mu x . \gamma$ is represented by $\mu\_.\_\,\ct\Fixpv\,\,\Exp\to\Exp$. (We chose not to provide any restriction to guarantee that $\gamma$ is a guarded expression, at this stage in the definition of $\mu\_.\_$. However, guards can be easily checked by pattern matching, according to the grammars in Definition~\ref{def:expr});

\item the specification of the substitution of a fixed-point variable with an expression, given by an operation
$\textbf{\_[\_\,/\_]} \ct \Exp~\Exp~\Fixpv \to \Exp$ and a set of equations -- one for each constructor. \ans{C.4.3} For example, the equations associated to $\emp$ and $\oplus$ are: $\emp[\eps/x] = \emp$, and respectively, $(\eps_1 \oplus \eps_2)[\eps/x] = (\eps_1[\eps/x]) \oplus (\eps_2[\eps/x])$, where $\eps, \eps_1, \eps_2$ are $\Gf$-expressions and $x$ is a fixed-point variable;

\item the specification of the type-checking relation in Definition~\ref{def:ts},
given by an operation 
$\_\,\ct\!\_ \,\,\ct \Exp~\Ingred\to {\sf Bool}$ and an equation for each inference rule defining this relation. For example the rule
\[
\rules{\vdash \E_1 \colon \F\lhd \Gf\;\;\;\; \vdash\E_2\colon \F\lhd \Gf}{\vdash \E_1\oplus\E_2 \colon \F\lhd \Gf}
\]
is represented by the equation
$
\E_1 \oplus \E_2 \ct \FtrlG =
\E_1 \ct \FtrlG \land \E_2 \ct \FtrlG
$. The type-checking operator is used in order to verify whether the expressions checked for equivalence are well-typed (Definition~\ref{def:g-expr}). Moreover, note that for the consistency of notation, algebraically we write
$\eps \ct \FtrlG$ to represent expressions $\eps$ of type $\FtrlG$.
\end{itemize}
\end{paragraph}

\medskip
\begin{paragraph}{\bf The algebraic specification of $\delta_\Gf$} It consists of:
\begin{itemize}
\item the specification of the coalgebra of $\Gf$-expressions $\delta_\Gf$ given by three operations
$\it \delta\_(\_) \ct \Ingred~\Exp\to \ExpS$, $\it Empty \ct \Ingred \to \ExpS$, and $\it Plus\_(\_,\_) \ct \Ingred~\ExpS~\ExpS\to \ExpS$;

\item a set of equations describing the definitions of these operations as in~\cite{brs_lmcs}.
\end{itemize}
\end{paragraph}

\medskip
\begin{paragraph}{\bf The algebraic specification of structured expressions}
As mentioned above, the set of $\Gf$-expressions is provided with a coalgebraic structure given by the function $\delta_{\Gf} \colon {\Exp}_{\Gf} \to {\Gf}({\Exp}_{\Gf})$, where ${\Gf}({\Exp}_{\Gf})$ can be understood as the set of expressions with structure given by $\Gf$ (and its ingredients). The set of structured expressions is defined by the following grammar:
\begin{equation}\label{eq:struct-expr}
\sigma \;::\!=\; \eps\mid b \mid \langle\sigma,\sigma\rangle\mid k_1(\sigma)\mid
               k_2(\sigma)\mid\bot\mid\top\mid\lambda {.}(a,\FtrlG,\sigma) \mid \{\sigma\}
\end{equation}
where $\eps \in \Exp_\Gf$ and $b\in \Bl$. The typing rules below give precise meaning to these expressions. Note that $\bot,\top$ are two expressions coming from $\Gf = \Gf_1 \myplus \Gf_2$, used to denote underspecification and overspecification, respectively.

The associated algebraic specification includes:
\begin{itemize}
\item a sort $\ExpS$ representing expressions $\sigma$ (from {$\F(\Exp_\Gf)$}, with $\FtrlG$), and one operation for each production in the BNF grammar~(\ref{eq:struct-expr}). Note that the construction $\lambda {.}(a,\FtrlG,\sigma)$ has as coalgebraic correspondent a function
$f \in \itF^{A}(\Exp_{\Gf})$, and is defined by cases as follows:
$\lambda {.}(a,\FtrlG,\sigma)(a')$ = {\it if} $(a=a')$ {\it then} $\sigma$ {\it else} ${\it Empty}_{\FtrlG}$;

\item the extension of the type-checking relation to structured expressions, defined by:
\\[0.5ex]
$
\begin{array}{@{\hspace{0.5cm}}l@{\hspace{0.5cm}}l}
\dfrac{\vdash b\ct \BtrlG}
       {\vdash b\in \Bl({\Exp}\,\Gf)}
&
\dfrac{\vdash \eps\ct \IdtrlG}
       {\vdash \eps\in {\id}({\Exp}\,\Gf)}
\\[3ex]
\dfrac{}
       {\vdash \bot\in \F_1{\myplus}\F_2({\Exp}\,\Gf)}
&
\dfrac{}
       {\vdash \top\in \F_1{\myplus}\F_2({\Exp}\,\Gf)}
\\[3ex]
\dfrac{\vdash \sigma\in \F_i({\Exp}\,\Gf)}
       {\vdash k_i(\sigma)\in \F_1{\myplus}\F_2({\Exp}\,\Gf)}~i\in\overline{1,2}
&
\dfrac{\vdash \sigma_1\in \F_i({\Exp}\,\Gf)\qquad\vdash \sigma_2\in \F_i({\Exp}\,\Gf)}
       {\vdash \langle\sigma_1,\sigma_2\rangle\in
\F_1{\times}\F_2({\Exp}\,\Gf)}\\[3ex]
\dfrac{\vdash \sigma\in \F({\Exp}\,\Gf),~a\in A}
       {\vdash \lambda.(a,\FtrlG,\sigma)\in \F^A({\Exp}\,\Gf)}
&
\dfrac{\vdash \sigma\in \F({\Exp}\,\Gf)}
       {\vdash \{\sigma\}\in \pow \F({\Exp}\,\Gf)}
\end{array}
$\\[1.5ex]
and specified by an operation $\_\in\!\_({\Exp}\,\_) \,\ct \it \ExpS~\Fun~\Fun\to {\sf Bool}$
(where we used a mix-fix notation) and an equation for
each of the above inference rules.
For example, the first rule has associated the equation $b\in \Bl({\Exp\,\Gf}) = b\ct \BtrlG$.
For consistency of notation, we write
$\sigma \in \F({\Exp}_\Gf)$ to denote that $\sigma$ is an element
of $\F({\Exp}_\Gf)$.
\end{itemize}
\end{paragraph}

\begin{remark}
In terms of membership equational logic (MEL)~\cite{Bouhoula-Jouannaud-Meseguer00},
both $\FtrlG$ and $\F({\Exp}\,\Gf)$ can be thought of as being sorts and,
for example, $\eps\ct\FtrlG$ as a membership assertion. Even if MEL is an elegant theory,
we prefer not to use it here because this implies the dynamic declaration of sorts and a set of assertions for such a sort.  The above approach is generic and therefore more flexible.
\end{remark}

\medskip
\begin{paragraph}{\bf The equational entailment relation $\vdashInd$ for bisimilarity checking}
As previously hinted in the beginning of this section, in order to algebraically reason on bisimilarity of $\Gf$-expressions in {\CIRC}, one has to extend the deduction relation $\vdash$ to allow a restricted contextual reasoning on expressions in $\F(\Exp_\Gf)$, for all ingredients $\F$ of a non-deterministic functor $\Gf$. We call the extended entailment $\vdashInd$.

The aforementioned restriction refers to inhibiting the use of congruence during equational reasoning, in order to guarantee the soundness of {\CIRC} proofs. This is realized by means of a \emph{freezing operator}, which intuitively behaves as a wrapper on the expressions checked for equivalence, by changing their sort to a fresh sort {\sf Frozen}. This way, the hypotheses collected during a {\CIRC} proof session cannot be used freely in contextual reasoning, hence preventing the derivation of untrue equations (as illustrated in Example~\ref{ex:streams}).

We further show how the freezing mechanism is implemented in our algebraic setting, and define $\vdashInd$.

Let ${\mathcal E}$ be an algebraic specification. We extend $\mathcal E$ by adding the freezing operation
$\efr{-} \ct s \rightarrow {\sf Frozen}$ for each sort $s \in \Sigma$,
where $\sf Frozen$ is a fresh sort.
By $\efr{t}$ we represent the \textit{frozen} form of a $\Sigma$-term $t$,
and by
$\efr{e}$ a \textit{frozen equation} of the shape
$(\forall X)\, \efr{t\phantom{\!'}} = \efr{t'} \textnormal{ if } c$. The entailment relation $\vdash$ is defined over frozen equations following the line in~\cite{rosu-lucanu-2009-calco}; \ans{C.1.11, C.2.5} more details are provided in Section~\ref{sec:dec-proced}.

\ans{C.1.17} Recall from Section~\ref{sec:prelim} that a relation
$\R \subseteq \Exp_\Gf \times \Exp_\Gf$ is a bisimulation if and only if $(s, t) \in \R \Rightarrow (\delta_{\GtrlG}(s), \delta_{\GtrlG}(t))\in \overline \Gf(\R)$. Here, $\overline{\Gf}(\R) \subseteq \Gf(\Exp_\Gf) \times \Gf(\Exp_\Gf)$ is the lifting of the relation $\R \subseteq \Exp_\Gf \times \Exp_\Gf$, defined as 

$$\overline \Gf(\R) = \{ (\Gf(\pi_1)(x),\Gf(\pi_2)(x)) \mid x \in \Gf(\R) \}\ .$$

 So, intuitively, reasoning on bisimilarity of two expressions $(\eps, \eps')$ in $\R$ reduces to checking whether the application of $\delta_\Gf$ maps them into $\overline{\Gf}(\R)$.
 
Therefore, checking whether a pair $(s^{\delta}, t^{\delta})$ is in $\overline{\Gf}(\R)$ consists in checking, for example for the case of $\Gf = \Gf_1 \times \Gf_2$, whether $(s_{1}^{\delta}, t_{1}^{\delta}) \in \overline{\Gf_1}(\R)$ and $(s_{2}^{\delta}, t_{2}^{\delta}) \in \overline{\Gf_2}(\R)$, where $s^{\delta} = <s^{\delta}_{1}, s^{\delta}_{2}>$ and $t^{\delta} = <t_{1}^{\delta}, t_{2}^{\delta}>$.
In an algebraic setting, this would reduce to building an algebraic specification $\algspecN$ and defining an entailment relation $\vdashInd$ such that one can infer $\algspecN \vdashInd \efr{<s^{\delta}_1, s^{\delta}_2>} = \efr{<t^{\delta}_1, t^{\delta}_2>}$ (this is the algebraic correspondent we consider for $(<s^{\delta}_1, s^{\delta}_2>, <t^{\delta}_1, t^{\delta}_2>) \in \overline{\Gf}(\R)$) by showing $\algspecN \vdashInd \efr{s^{\delta}_1} = \efr{t^{\delta}_1}$ (or $(s^{\delta}_1, t^{\delta}_1) \in \overline{\Gf_1}(\R)$) and 
$\algspecN \vdashInd \efr{s^{\delta}_2} = \efr{t^{\delta}_2}$ (or $(s^{\delta}_2, t^{\delta}_2) \in \overline{\Gf_2}(\R)$).
We hint that the aforementioned algebraic specification $\algspecN$ consists of $\algspec$ and a set of frozen equations (see Corollary~\ref{cor:ii}).

The entailment relation $\vdashInd$ for reasoning on bisimilarity of $\Gf$-expressions is based on the definition of $\overline{\Gf}$.
%
%
\begin{definition}
\label{def:PF}
The entailment relation $\vdashInd$ is the extension of $\vdash$ with the following inference rules,
which allow a  restricted contextual reasoning over the frozen equations of structured expressions:
\begin{equation}
\frac{\algspec  \vdashInd \efr{\sigma_1 \phantom{\hspace{-2.5ex}\sigma_{1}'}} = \efr{\sigma'_1} \,\,\,\,\,\,
\algspec  \vdashInd \efr{\sigma_2\phantom{\hspace{-2.5ex}\sigma_{2}'}} = \efr{\sigma_2'}}
{\algspec  \vdashInd \efr{\langle \sigma_1, \sigma_2\rangle} =
\efr{\langle \sigma_1', \sigma_2'\rangle}}
\label{rl:times}
\end{equation}
\vspace{-1ex}
\begin{equation}
\frac{\algspec  \vdashInd \efr{\sigma\phantom{\hspace{-2.1ex}\sigma'}} = \efr{\sigma'}}
{\algspec  \vdashInd \efr{k_i(\sigma)} = \efr{k_i(\sigma')}}~i \in \overline{1,2}
\label{rl:plus}
\end{equation}
\begin{equation}
\frac{\algspec  \vdashInd \efr{f(a)} = \efr{g(a)}\, , ~{\it for~all~} a \in {\it A}}
{\algspec  \vdashInd \efr{f} = \efr{g\phantom{\hspace{-1.1ex}f}}}\label{rl:expo}
\end{equation}
\begin{equation}
\frac{\algspec  \vdashInd \efr{\sigma_{i_1}\phantom{\hspace{-3.0ex}\sigma'_{j_1}}} = \efr{\sigma'_{j_1}}\,,
\ldots,\,
\algspec  \vdashInd \efr{\sigma_{i_k}\phantom{\hspace{-3.4ex}\sigma'_{j_k}}} = \efr{\sigma'_{j_k}}}
{\algspec  \vdashInd \efr{\{\sigma_1, \ldots, \sigma_n\}} = \efr{\{\sigma'_1, \ldots, \sigma_m'\}}}{\small
\begin{array}{l} \{i_1,\ldots, i_k\} = \{1,\ldots, n\}\\  \{j_1,\ldots, j_k\} = \{1,\ldots, m\}\end{array}
}\label{rl:pow}
\end{equation}

\end{definition}

\begin{remark}
Note that the extension of the entailment relation $\vdash$ to $\vdashInd$ implies that
${\algspec  \vdash e} \textnormal{\it{ iff }} {\algspec  \vdashInd e}$ holds, for any equation $e$ of shape $\,\efr{\eps_1} = \efr{\eps_2}$ or $\eps_1 = \eps_2$, with $\eps_1, \eps_2$ non-structured expressions. Below, we will use the notation $\algspec  \vdashInd \R$, where $\R$ is a set of possibly frozen equations, to denote $\forall_{e\in\R}\cdot\algspec  \vdashInd e$. 
\end{remark}

It is interesting to recall the relation lifting for the powerset functor which is encoded in the last rule of Definition~\ref{def:PF}. A pair $(U,V)$ is in  $\overline{\pow\Gf}(\R)$ if and only if for every $u\in U$ there exists a $v\in V$ such that $(u,v)$ belongs to $\overline\Gf(\R)$ and, conversely, for every $v\in V$, there exists a $u\in U$ such that $(u,v)$ belongs to $\overline\Gf(\R)$.

\begin{remark}
\ans{C.2.9}
As already hinted (and proved in Corollary~\ref{cor:ii}), reasoning on bisimilarity of expressions in a binary relation $\R \subseteq \Exp_\Gf \times \Exp_\Gf$ reduces to showing that 
$\efr{\delta_\Gf(s)} = \efr{\delta_\Gf(t)}$ is a $\vdashInd$-consequence, for all $(s,t) \in \R$.
The equational proof is performed in a ``top-down" fashion, by reasoning on the subsequent equalities between the components of the corresponding structured expression $\delta_\Gf(s)$, $\delta_\Gf(t)$ in an inductive manner. This is realized by applying the inverted rules (\ref{rl:times})--(\ref{rl:pow}).

Moreover, note that rule (\ref{rl:pow}) is not invertible in the usual sense; rather any statement matching the form of the conclusion can only be proved by some instance of the rule.
\end{remark}
\end{paragraph}


We will further formalize the connection between the inductive definition of $\overline{\Gf}$ (on the coalgebraic side)
and $\vdashInd$ (on the algebraic side) in Theorem~\ref{thm:i}, hence enabling the
definition of bisimulations in algebraic terms,
in Corollary~\ref{cor:ii}.

\begin{remark}\label{rem:eqRed}
Equations in $\algspec$ (built as previously described in this section) are used in the equational reasoning only for reducing terms of shape ${\sf op}(t_1, \ldots, t_n)$ according to the definition of the operation ${\sf op}$. For the simplicity of the proofs of Theorem~\ref{thm:i} and Corollary~\ref{cor:ii}, whenever we write ${\sf op}(t_1, \ldots, t_n)$, we refer to the associated term reduced according to the definition of ${\sf op}$.
\end{remark}

First we introduce some notation conventions.
Let $\Gf$ be a non-deterministic functor and $\R \subseteq \Exp_{\Gf} \times \Exp_{\Gf}$.
We write:
\begin{itemize}\itemsep1pt
\item $\RId$ to denote the set $\RG \cup \{(\eps, \eps) \mid\algspec \vdash \eps \ct \GtrlG = {\it true}\}$;
\item $\cl{\R}$ for the closure of $\R$ under transitivity, symmetry and reflexivity;
\item $\efr{\mathcal R}$ to represent the set $\bigcup_{e \in {\mathcal R}}\{\efr{e}\}$;
(application of the freezing operator to all elements of $\R$)
\item $\delta_{\GtrlG}(\eps=\eps')$ to represent the equation $\delta_{\GtrlG}(\eps)=\delta_{\GtrlG}(\eps')$;

\item $\algspec \cup \efr{\R}$ as a shorthand for
$(S, \Sigma, E \cup \{\efr{\E\phantom{\hspace{-1.5ex}\E'}} = \efr{\E'} \mid (\E, \E') \in \R \})$, where $\algspec = (S, \Sigma, E)$;

\item \ans{C.2.6} $(\sigma, \sigma') \in  \overline\Gf(\RG)$ as a shorthand for: $(\sigma, \sigma')$ is among the enumerated elements of a set $S$ explicitly constructed as an enumeration of the finite set $\overline\Gf(\RG)$ (in the algebraic setting, $\overline\Gf(\RG)$ is a subset of ${\mathcal T}_{\Sigma, \ExpS\!\!} \times {\mathcal T}_{\Sigma, \ExpS\!\!}$ and $\algspec \vdash \overline\Gf(\RG) = S$).

\end{itemize}

\begin{theorem}\label{thm:i}
Consider a non-deterministic functor $\Gf$. Let $\F$ be an ingredient of $\Gf$, $\RG$ a binary relation on
the set of $\Gf$-expressions,
and $\sigma, \sigma' \in {\F}(\Exp_{\itG})$.
\begin{itemize}\itemsep2pt
\item[a)] If $\Gf$ is not a constant functor, then
$(\sigma, \sigma') \in \overline{\F}(\cl{\RId})$ iff
$\algspec \cup \efr{\RG} \vdashInd \efr{\sigma} = \efr{\sigma'}$;

\item[b)] If $\Gf$ is a constant functor $\Bl$, then
$(\sigma, \sigma') \in \overline{\Bl}(\cl{\RId})$ iff
$\algspec \vdashInd \efr{\sigma\phantom{\hspace{-1.7ex}\sigma'}} = \efr{\sigma'}$.
\end{itemize}
\end{theorem}

\noindent
In order to prove Theorem~\ref{thm:i}.$a)$ we introduce the following lemma:

\begin{lemma}\label{lm:idir}
Consider $\Gf$ a non-deterministic functor and
$\RG$ a binary relation on
the set of $\Gf$-expressions.
If $(\eps, \eps') \in \cl{\RId}$ then
$\algspec \cup \efr{\RG} \vdashInd \efr{\eps\phantom{\hspace{-1.7ex}\eps'}} = \efr{\eps'}$.
\end{lemma}

\begin{proof}
\ans{C.2.7}
The proof is trivial, as equality is reflexive, symmetric and transitive.
\qed
\end{proof}

We are now ready to prove Theorem~\ref{thm:i}.
\begin{proof}[Theorem~\ref{thm:i}]
\verb##
\begin{itemize}
\item Proof of Theorem~\ref{thm:i}.$a)$.\\
\begin{itemize}\itemsep10pt
\item {$``\Rightarrow"$.}
The proof is by induction on the structure of $\itF$.\\
\textit{Base case}:
\begin{itemize}\itemsep6pt
\item $\itF = \itB$. It follows that $(\sigma,\sigma')$ is of shape
$(b, b)$ where $b\in \itB$, therefore
$\algspec \cup \efr{\RG} \vdashInd \efr{b} = \efr{b}$ holds
by reflexivity.

\item $\itF = \id$. In this case $(\sigma, \sigma') \in cl(\RId) = \overline{\id}(cl(\RId))$,
so the result follows immediately by Lemma~\ref{lm:idir}.
\end{itemize}
\textit{Induction step}:
\begin{itemize}\itemsep6pt
\item ${\itF} = {\itF}_1 \times {\itF}_2$. Obviously,
$\sigma = < \sigma_1, \sigma_2>$ and $\sigma' = < \sigma'_1, \sigma'_2>$,
where $(\sigma_1, \sigma'_1) \in \overline{{\itF}_1}(cl(\RId))$
and $(\sigma_2, \sigma'_2) \in \overline{{\itF}_2}(cl(\RId))$.
Therefore, by the induction hypothesis, both
$\algspec \cup \efr{\RG} \vdashInd \efr{\sigma_1\phantom{\hspace{-2.3ex}\sigma'_{1}}} = \efr{\sigma'_1}$
and
$\algspec \cup \efr{\RG} \vdashInd \efr{\sigma_2\phantom{\hspace{-2.3ex}\sigma'_{2}}} = \efr{\sigma'_2}$ hold.
Hence, according to the definition of $\vdashInd$ (see~(\ref{rl:times})),
we conclude that
$\algspec \cup \efr{\RG} \vdashInd \efr{<\sigma_1, \sigma_2>} = \efr{<\sigma'_1, \sigma'_2>}$ holds.

\item The cases
${\itF} = {\itF}_1 \myplus {\itF}_2$, ${\itF} = {\itF}_1^{\it A}$ and ${\itF} = \pow{\itF'}$
are handled in a similar way.
\end{itemize}

\item {$``\Leftarrow"$.}
We proceed also by induction on the structure of $\itF$. Moreover, recall that the observations in Remark~\ref{rem:eqRed} hold (for each of the subsequent cases).
\\
\textit{Base case}:
\begin{itemize}\itemsep6pt
\item $\itF = \Bl$. In this case $(\sigma, \sigma')$ is of shape
$(b, b')$, where $b, b'$ are two elements of the
semilattice $\Bl$.
Also, recall that $\Gf \not = \Bl$, therefore, the equations (of type $\GtrlG \not = \F(\Exp_\Gf)$) in $\R$ are not involved in the equational reasoning. We deduce that $\efr{b\phantom{\hspace{-1.3ex}b'}} = \efr{b'}$
is proved by reflexivity, hence $(b, b') = (b, b) \in \overline{\Bl}(\cl{\RId})$.

\item $\itF = \id$. Note that for this case, $\sigma, \sigma'$ are expressions of the same type with the expressions in $\R$. We further identify two possibilities:
\begin{itemize}\itemsep2pt
\item $\efr{\sigma\phantom{\hspace{-1.7ex}\sigma'}} = \efr{\sigma'}$ is proved by reflexivity,
therefore $(\sigma, \sigma') \in \{(\eps, \eps) \mid \eps : \GtrlG\} \subseteq
\RId \subseteq \cl{\RId} = \overline{\id}(\cl{\RId})$.

\item the equations in
$\efr{\RG}$ are used in the equational reasoning
$\algspec \cup \efr{\RG} \vdashInd \efr{\sigma\phantom{\hspace{-1.7ex}\sigma'}} = \efr{\sigma'}\,$.
In addition, the freezing operator inhibits contextual reasoning, therefore $\efr{\sigma\phantom{\hspace{-1.7ex}\sigma'}} = \efr{\sigma'}$ is proved according to the equations in $\efr{\R}$, based on the symmetry and transitivity of $\vdashInd$.
In other words, $(\sigma, \sigma') \in \cl{\RId} = \overline{\id}(\cl{\RId})$.

\end{itemize}
\end{itemize}
\textit{Induction step}:
\begin{itemize}\itemsep6pt
\item ${\itF} = {\itF}_1 \times {\itF}_2$.
Obviously, due to their type, the equations in $\R$ are not involved in the equational reasoning. Also, recall that (*) holds.
Therefore, $\algspec \cup \efr{\RG} \vdashInd \efr{<\sigma_1, \sigma_2>} = \efr{<\sigma'_1, \sigma'_2>}$ is a consequence of the inverted rule~(\ref{rl:times}). More explicitly, it follows that 
$\algspec \cup \efr{\RG} \vdashInd \efr{\sigma_1\phantom{\hspace{-2.3ex}\sigma'_{1}}} = \efr{\sigma'_1}$
and $\algspec \cup \efr{\RG} \vdashInd \efr{\sigma_2\phantom{\hspace{-2.3ex}\sigma'_{2}}} = \efr{\sigma'_2}$ must hold.
By the induction hypothesis, we deduce that
$(\sigma_1, \sigma'_1) \in \overline{\F}_1(\cl{\RId})$
and $(\sigma_2, \sigma'_2) \in \overline{\F}_2(\cl{\RId})$.
So by the definition of $\overline{{\itF}_1 \times {\itF}_2}$
we conclude that
$(\langle \sigma_1, \sigma_2 \rangle, \langle \sigma'_1, \sigma'_2 \rangle) =
(\sigma, \sigma') \in \overline{{\itF}_1 \times {\itF}_2}(\RG)$.

\item The cases
${\itF} = {\itF}_1 \myplus {\itF}_2$, ${\itF} = ({\itF}_1)^{\it A}$
and ${\itF} = \pow{\itF'}$
follow a similar reasoning.
\end{itemize}
\end{itemize}

\item Proof of Theorem~\ref{thm:i}.$b)$. It follows immediately by the definition of $\overline{\Bl}$ and Remark~\ref{rem:eqRed}.
\end{itemize}
\qed
\end{proof}

\begin{corollary}\label{cor:ii}
Let $\Gf$ be a non-deterministic functor and 
$\RG$ a binary relation on the set of $\Gf$-expressions.
\begin{itemize}\itemsep2pt
\item[a)] If $\Gf$ is not a constant functor, then
$\cl{\RId}$ is a bisimulation iff
$\algspec \cup \efr{\RG} \vdashInd \efr{\delta_{\GtrlG}(\RG)}$;

\item[b)] If $\Gf$ is a constant functor $\Bl$, then
$\cl{\RId}$ is a bisimulation iff
$\algspec \vdashInd \efr{\delta_{\GtrlG}(\RG)}$.
\end{itemize}
\end{corollary}

\begin{proof}
\verb##
\begin{itemize}\itemsep2pt
\item Proof of Corollary~\ref{cor:ii}.$a)$. We reason as follows:
\[
\begin{array}{lclr}
&&\cl{\RId} \text{ is a bisimulation } \\[2ex]&\Leftrightarrow& (\forall (\eps, \eps') \in \cl{\RId}).((\delta_{\GtrlG}(\eps),
\delta_{\GtrlG}(\eps')) \in \overline{\mathscr G}(\cl{\RId}) & {(\textnormal{Def.~\ref{def:bis}})}\\[2ex]&\Leftrightarrow& \algspec \cup \efr{\RG} \vdashInd \efr{\delta_{\GtrlG}(\cl{\RId})} & {(\textnormal{Thm.~\ref{thm:i}})}\\[2ex]
&\Leftrightarrow&
\algspec \cup \efr{\RG} \vdashInd \efr{\delta_{\GtrlG}(\RG)}&{(\cl{\RId}, \vdashInd)}
\end{array}
\]

\item Proof of Corollary~\ref{cor:ii}.$b)$. It follows immediately by the definition of bisimulation relations and according to the observations in Remark~\ref{rem:eqRed}.
\end{itemize}
\qed
\end{proof}


In Figure~\ref{fig:CoVsAlg} we briefly summarize the results of the current section, namely, the algebraic encoding of the coalgebraic setting presented in~\cite{brs_lmcs}.

\begin{figure}[h]
\centering
\renewcommand{\arraystretch}{1.3}
\begin{tabular}{|c|cr|}
\hline
coalgebraic & algebraic &\\
\hline\hline

$\vdash \eps \ct \FtrlG$ & $\algspec \vdash \eps \ct \FtrlG = {\it true}$&\\
\hline

${\Exp_{\FtrlG}}$
&
\multicolumn{2}{|c|}{$\{\eps \in {\mathcal T}_{\Sigma, \Exp\!\!} \mid \algspec \vdash \eps \ct \FtrlG = {\it true}\}$}\\
\hline

${\Exp_{\Gf}}$
&
\multicolumn{2}{|c|}{$\{\eps \in {\mathcal T}_{\Sigma, \Exp\!\!} \mid \algspec \vdash \eps \ct \GtrlG = {\it true}\}$}\\
\hline

${\F(\Exp_\Gf)}$ &
\multicolumn{2}{|c|}{$\{\sigma \in {\mathcal T}_{\Sigma, \ExpS\!\!} \mid
\algspec \vdash \sigma\in \F({\Exp}\,\Gf) = {\it true} \}$}\\
\hline

$\delta_{\FtrlG} \ct {\Exp_{\FtrlG}} \rightarrow {\F(\Exp_\Gf)}$ &
\multicolumn{2}{|c|}{
$\delta \_ ( \_ ) \ct {\Ingred~\Exp \rightarrow \ExpS}$}\\
\hline

&
$
\begin{array}{c}
\algspec \vdash \sigma\in \F({\Exp}\,\Gf) = {\it true},\\
\algspec \vdash \sigma'\in \F({\Exp}\,\Gf) = {\it true}
\end{array}
$
&\\

$(\sigma, \sigma') \in \overline{\F}(\cl{\RId})$ &
$\algspec \cup \efr{\R} \vdashInd \efr{\sigma\phantom{\hspace{-1.8ex}\sigma'}} = \efr{\sigma'}$\,\, if $\Gf \not= \Bl$ &\\

& or &\\

& $\algspec \vdashInd \efr{\sigma\phantom{\hspace{-1.8ex}\sigma'}} = \efr{\sigma'}$\,\, if $\Gf = \Bl$ & (Thm.~\ref{thm:i})\\

\hline


$\cl{\RId}$ is a bisimulation &
$\algspec \cup \efr{\R} \vdashInd \efr{\delta_{\GtrlG}(\R)}$\,\, if $\Gf \not= \Bl$&\\

& or &\\

& $\algspec \vdashInd \efr{\delta_{\GtrlG}(\R)}$\,\, if $\Gf = \Bl$&(Cor.~\ref{cor:ii})\\
\hline
\end{tabular}
\caption{non-deterministic functors - coalgebraic vs. algebraic approach}
\label{fig:CoVsAlg}
\end{figure}

\section{A Decision Procedure for Bisimilarity in \CIRC}
\label{sec:dec-proced}

In this section, we describe how the coinductive theorem prover {\CIRC}~\cite{lucanu-etal-2009-calco} can be used to implement the decision procedure for the bisimilarity of generalized regular expressions, which we discussed above.

{\CIRC} can be seen as an extension of Maude
with behavioral features and its implementation
is derived from that of Full-Maude.
In order to use the prover, one needs to provide a specification (a {\CIRC} theory)
and a set of goals.
A {\CIRC} theory ${\cal B}=(S, (\Sigma,\Delta),(E,{\cal I}))$ consists of an algebraic specification $(S,\Sigma,E)$, a set $\Delta$ of \emph{derivatives}, and a set $\cal I$ of equational interpolants, which are expressions of the form $e \Rightarrow \{ e_i \mid i \in I \}$ where $e$ and $e_i$ are equations.
\ans{C.3.26}
The intuition for this type of expressions is simple: $e$ holds whenever for any $i$ in $I$ the equation $e_i$ holds. In other words, to prove $E \vdash e$ one can chose to instead prove $E \vdash \{e_i \mid i \in I\}$.  For the particular case of non-deterministic functors, we use equational interpolants to extend the initial entailment relation in a consistent way with rules~(\ref{rl:times})--(\ref{rl:pow}).
(For more information on equational interpolants see \cite{acca}). A derivative $\delta\in\Delta$ is a $\Sigma$-term containing a special variable ${*}{:}s$ \ans{C.1.20} ({\it i.e.}, a $\Sigma$-context), where $s$ is the sort of the variable $*$. If $e$ is an equation $t=t'$ with $t$ and $t'$ of sort $s$, then $\delta[e]$ is $\delta[t/{*}{:}s]=\delta[t'/{*}{:}s]$.
We call this type of equation a \emph{derivable equation}. The other equations are \emph{non-derivable}. We write $\delta[\R]$ to represent $\{\delta[e] \mid e \in \R\}$, where $\R$ is a set of derivable equations, and $\Delta[e]$ for the set $\{\delta[e]\mid\delta\in\Delta{\rm ~appropriate~for~}e\}$.

\ans{C1.11} Moreover, note that {\CIRC} works with an extension of the entailment relation $\vdash$ over frozen equations (introduced in Section~\ref{sec:algSpec}), with two more axioms, as in~\cite{rosu-lucanu-2009-calco}:
\begin{equation}
E \cup \R \vdash \efr{e}\,\,\,\, \textnormal{\it{iff}}\,\,\, E \vdash e \label{A1}\\
\end{equation}
\begin{equation}
E \cup \R \vdash \G\,\,\,\, \it{implies}\,\,\, E \cup \delta[\R] \vdash \delta[\G] \textnormal{\it{ for each }} \delta \in \Delta \label{A2}
\end{equation}

Above, $E$ ranges over unfrozen equations, $e$ over non-derivable unfrozen equations, and $\R,\G$ over derivable frozen equations.

\begin{remark}
Note that the new entailment $\vdashInd$ extended over frozen equations (in Definition~\ref{def:PF}) satisfies the assumptions~(\ref{A1}) and~(\ref{A2}).
\end{remark}

{\CIRC} implements the coinductive proof system
given in~\cite{rosu-lucanu-2009-calco} using a set of reduction rules
of the form
$({\cal B}, {\cal F}, {\cal G}) \Rightarrow ({\cal B}, {\cal F'}, {\cal G'})$, where
${\cal B}$ represents a specification,
${\cal F}$ is the coinductive hypothesis (a set of frozen equations) and
${\cal G}$ is the current set of goals. The freezing operator is defined
as described in Section~\ref{sec:algSpec}.
Here is a brief description of these rules:
\begin{itemize}\itemsep 5pt
\item[]
\textsf{[Done]}:
$({\cal B}, {\cal F}, \{\})
\Rightarrow
\cdot
$\\
Whenever the set of goals is empty, the system terminates with success.
\item[]
\textsf{[Reduce]}:
$({\cal B}, {\cal F}, {\cal G} \cup \{\efr{e}\})
\Rightarrow
({\cal B}, {\cal F}, {\cal G})
\textit{\;if\;} {\cal B} \cup {\cal F} \vdash \efr{e}
$\\
If the current goal is a $\vdash$-consequence
of ${\cal B} \cup {\cal F}$ then $\efr{e}$ is removed from
the set of goals.
\item[]
\textsf{[Derive]}:
$({\cal B}, {\cal F}, {\cal G} \cup \{\efr{e}\})
\Rightarrow
({\cal B}, {\cal F} \cup \{\efr{e}\}, {\cal G} \cup \efr{\Delta[e]})$
$\textit{if\;} {\cal B} \cup {\cal F}\, {\not \vdash}\, \efr{e} $\\
When the current goal $e$ is derivable and it is not a
$\vdash$-consequence, it is added
to the hypothesis and its derivatives to the set of goals.
\ans{C.1.21}

\item[]
\textsf{[Simplify]}:
$({\cal B}, {\cal F}, {\cal G} \cup \{\efr{\theta(e)}\})
\Rightarrow
({\cal B}, {\cal F}, {\cal G} \cup \{\efr{\theta(e_i)} \mid i \in I\})$
\\
\verb##{\hspace{38pt}}
$\textit{\;if\;}  e  \Rightarrow \{ e_i \mid i \in I \}$
is an equational interpolant
from the\\
\verb##{\hspace{41pt}}
specification and $\theta \ct X \rightarrow {\mathcal T}_{\Sigma}(Y)$ is a substitution. \ans{C.2.12}
\item[]
\textsf{[Fail]}:
$({\cal B}, {\cal F}, {\cal G} \cup \{\efr{e}\})
\Rightarrow
{\it failure} \textit{\;if\;} {\cal B} \cup {\cal F} \, {\not \vdash}\, \efr{e} \,\land\,
e$ {\it{is non-derivable}}\\
This rule stops the reduction process  with failure whenever the current goal
$e$ is non-derivable and is not a $\vdash$-consequence of ${\cal B} \cup {\cal F}$ \ans{C.2.13}.
\end{itemize}

It is worth noting that there is a strong connection between a {\CIRC} proof and the construction of a bisimulation relation. We illustrate this fact and the importance of the freezing operator with a simple example.

\begin{example}\label{ex:streams}
Consider the case of infinite streams. The set $\itB^\omega$ of infinite streams over a set $\itB$ is the final coalgebra of the functor $\Rf = \it \itB \times \id$, with a coalgebra structure given by {\it hd} and {\it tl}, the functions that return the head and the tail of the stream, respectively. Our purpose is to prove that $0^\infty = (00)^\infty$. Let $z$ and $zz$ represent the stream on the left hand side and, respectively, on the right hand side. These streams are defined by the equations: ${\it hd}(z) = 0, {\it tl}(z) = z, {\it hd}(zz) = 0, {\it tl}(zz) = 0:zz$. 
Note that equations over $\itB$ like ${\it hd}(z) = 0$ are not derivable and equations over streams like ${\it tl}(z) = z$ are derivable.

In Fig.~\ref{fig:parallel} we present the correlation between the {\CIRC} proof and the construction of the bisimulation relation. Note how {\CIRC} collects the elements of the bisimulation as frozen hypotheses.

\begin{figure}[h]
\centering
\renewcommand{\arraystretch}{1.9}
\begin{tabular}{|c|c|}
\hline
{\CIRC} proof & Bisimulation construction \\
\hline\hline

\raisebox{-7pt}{\code{(add goal z = zz .)}} &
$\xymatrix@R=.2cm@C=1cm{
*+[o][F]{z}\ar@(dr,ur) \ar@{=>}[d] &
*+[o][F]{zz}\ar@/^/[r] \ar@{=>}[d]&
*+[o][F]{(zz)'}\ar@/^/[l] \ar@{=>}[d]\\
0 & 0 & 0
}$\\
\hline
$({\cal B}, \{\}, \{\efr{z} = \efr{zz}\})$ &
${\cal F} = \{\}; \,\, z \sim zz~?$
\\
\hline

\raisebox{-4pt}{$
\mathrel{{\overset{\textsf{[Derive]}}{\longrightarrow}}}
\left({\cal B}, \{\efr{z} = \efr{zz}\},
\left\{{\efr{{\it hd}(z)} = \efr{{\it hd}(zz)}} \atop
{\efr{{\it tl}(z)} = \efr{{\it tl}(zz)}}\right\}\right)
$}
&
\raisebox{-4pt}
{$
{\cal F} = \{(z, zz)\};
{
\,\, z \mathrel{{\overset{0}{\longrightarrow}}} z
\atop
\,\, zz \mathrel{{\overset{0}{\longrightarrow}}} (zz)'
}
$}
\\
\hline
$
\mathrel{{\overset{\textsf{[Reduce]}}{\longrightarrow}}}
({\cal B}, \{\efr{z} = \efr{zz}\},
\{\efr{z} = \efr{0 : zz}\})
$
&
{
$
{\cal F} = \{(z, zz)\};
\,\, z \sim (zz)'~?
$
}
\\
\hline

\raisebox{-4pt}{$
\mathrel{{\overset{\textsf{[Derive]}}{\longrightarrow}}}
\left({\cal B},
\left\{{\efr{z} = \efr{zz}} \atop
{\efr{z} = \efr{0 : zz}}\right\}
,
\left\{{\efr{{\it hd}(z)} = \efr{{\it hd}(0:zz)}} \atop
{\efr{{\it tl}(z)} = \efr{{\it tl}(0:zz)}}\right\}\right)
$}
&
\raisebox{-3pt}{$
{\cal F} = \{(z, zz), (z, (zz)')\};
{
\,\, z \mathrel{{\overset{0}{\longrightarrow}}} z
\atop
\,\, (zz)' \mathrel{{\overset{0}{\longrightarrow}}} zz
}
$}
\\
\hline
\raisebox{-3pt}{$\mathrel{{\overset{\textsf{[Reduce]}}{\longrightarrow}}}
\left({\cal B},
\left\{{\efr{z} = \efr{zz}} \atop
{\efr{z} = \efr{0 : zz}}\right\}
,
\{\}\right)
$}
&
$
{\cal F} = \{(z, zz), (z, (zz)')\}~\checkmark
$
\\
\hline
\end{tabular}
\caption{Parallel between a {\CIRC} proof and the bisimulation construction}
\label{fig:parallel}
\end{figure}

Let us analyze what would happen if the freezing operator $\efr{-}\,$ were not used. Suppose the circular coinduction algorithm would add the equation $z = zz$ in its unfrozen form to the hypotheses. After applying the derivatives we obtain the goals
${\it hd}(z) = {\it hd}(zz), {\it tl}(z) = {\it tl}(zz)$.
At this point, the prover could use the freshly added equation \ans{C.1.22} $z = zz$, and according to the congruence rule, both goals would be proven directly, though we would still be in the process of showing that the hypothesis holds. By following a similar reasoning, we could \ans{C.1.23} also prove that $0^\infty = 1^\infty$! In order to avoid these situations, the hypotheses are frozen, ({\it i.e.}, their sort is changed from {\sf Stream} to {\sf Frozen}) and this stops the application of the congruence rule, forcing the application of the derivatives according to their definition in the specification. Therefore, the use of the freezing operator is vital for the soundness of circular coinduction.

\end{example}

Next, we focus on using {\CIRC} for automatically reasoning on the equivalence of $\Gf$-expressions. As we will show, the implementation of both
the algebraic specifications associated to non-deterministic functors
and the equational entailment relation described
in Section~\ref{sec:algSpec} is immediate.
Given a non-deterministic functor $\Gf$, we define a {\CIRC} theory
$\behspec=(S, (\Sigma,\Delta),(E,{\cal I}))$ as follows:\\[-4ex]
\begin{itemize}
\item $(S,\Sigma,E)$ is $\algspec$
\item $\Delta=\{\delta_{\GtrlG}(*{:}\Exp)\}${, so the only derivable equations are those of sort \texttt{Exp}. As we have already seen for the example of streams, equations of sort \texttt{Slt} must not be derivable. Since we have the subsort relation $\tt Slt \ls  Exp$, we avoid the application of the derivative $\delta_{\GtrlG}(*{:}\Exp)$ over equations of sort \texttt{Slt} by means of an interpolant (see below).}
\item $\cal I$ consists of the following equational interpolants \ans{C.1.25}, whose role is to replace current proof obligations 
{over non-trivial structures with simpler ones:}
\begin{align}
\langle  \sigma_1, \sigma_2  \rangle  =  \langle  \sigma'_1, \sigma'_2  \rangle \,\, \Rightarrow \,\, &
\,\{ \sigma_1 = \sigma'_1,\,\, \sigma_2 = \sigma'_2 \} \label{srl:times}\\
k_i(\sigma) = k_i(\sigma') \,\, \Rightarrow \,\, & \,\{\sigma = \sigma'\} \label{srl:plus}\\
 f = g \,\, \Rightarrow \,\, & \,\{ f(a) = g(a)\mid a \in A\} \label{srl:expo}\\
 \cup_{i\in\overline{1,n}}\{\sigma_i\} = \cup_{j\in\overline{1,m}}\{\sigma'_j\}
\,\, \Rightarrow \,\, & \,\{ \land_{i\in\overline{1,n}}(\lor_{j\in\overline{1,m}}\, \sigma_i = \sigma'_j)\notag\\
& \hspace{4.5pt}\land_{j\in\overline{1,m}}(\lor_{i\in\overline{1,n}}\, \sigma_i = \sigma'_j)\}
\label{srl:pow}
\end{align}
{
together with an equational interpolant
\begin{align}
& t=t' \,\, \Rightarrow \,\, \{t\simeq t'=\tt true\}\label{srl:slt}
\end{align}
where $\simeq$ is the equality predicate equationally defined over the sort \texttt{Slt}. The last interpolant transforms the equations of sort \texttt{Slt} from derivable (because of the subsort relation $\tt Slt\ls Exp$) into non-derivable and equivalent ones.\\ 
}
\end{itemize}


The interpolants 
{(\ref{srl:times}--\ref{srl:slt})}
in $\cal I$ extend the entailment relation $\vdashInd$ (introduced in Definition~\ref{def:PF}) as follows:
\[
\dfrac{E\vdashInd \{e_i\mid i\in I\}}
         {E\vdashInd e}~{\rm if~}e\Rightarrow\{e_i\mid i\in I\}{\rm~in~}{\cal I}
\]

\begin{theorem}[Soundness]\label{thm:soundnessCirc}
Let $\Gf$ be a non-deterministic functor, and $\G$ a binary relation on the
set of $\mathscr G$-expressions.\\
If $({\behspec}, {\mathcal F}_0 = \{\}, \G_0 = \efr{\G}) \rTrans{*}
({\behspec}, {\mathcal F}_n, \G_n = \{\})$
using \textsf{[Reduce]}, \textsf{[Derive]} and \textsf{[Simplify]},
then $\G \subseteq \sim_{\Gf}$.
\end{theorem}

\begin{proof}
The idea of the proof is to find a bisimulation relation
$\widetilde{\cal F}$ s.t. $\G
\subseteq \widetilde{\cal F}$.\\
First let $\mathcal{F}$ represent the set of hypotheses (or derived goals) collected during the proof session.
We distinguish between two cases:

\begin{itemize}\itemsep6pt
\item[a)] $\Gf = \Bl$.
For this case, the set of expressions in $\G$ is given by the following grammar:
\begin{equation}
\eps\,::= \emp \mid b \mid \eps \oplus \eps \mid \mu x . \eps\,.
\label{eq:const}
\end{equation}

{
Note that the goals $\eps = \eps'$ in $\G$ are proven
\begin{enumerate}
\item either according to \textsf{[Simplify]}, applied in the context of the equational interpolant~(\ref{srl:slt}). If this is the case, then $\eps = \eps'$ holds by reflexivity, therefore
\begin{equation}\label{eq:s}
\behspec \vdashInd \efr{\delta_{\Bl \triangleleft \Bl}(\eps)} = \efr{\delta_{\Bl \triangleleft \Bl}(\eps')}
\end{equation}
also holds;
\item or after the application of \textsf{[Derive]}, case in which $\behspec \cup \efr{\cal F} \vdashInd \efr{\delta_{\Bl \triangleleft \Bl}(\eps)} = \efr{\delta_{\Bl \triangleleft \Bl}(\eps')}$ holds. Moreover, note that $\delta_{\Bl \triangleleft \Bl}(\eps)$ and $\delta_{\Bl \triangleleft \Bl}(\eps')$ are reduced to $b$, respectively $b' \in \Bl$, according to~(\ref{eq:const}) and the definition of $\delta_{\Bl \triangleleft \Bl}$. Consequently, the non-derivable (due to the subsort relation $\Bl\ls \tt Slt$) goal $\efr{b} = \efr{b'}$ holds by reflexivity, so the following is a sound statement:
\begin{equation}\label{eq:ss}
\behspec \vdashInd \efr{\delta_{\Bl \triangleleft \Bl}(\eps)} = \efr{\delta_{\Bl \triangleleft \Bl}(\eps')}.
\end{equation}
\end{enumerate}
Based on~(\ref{eq:s}),~(\ref{eq:ss}) and Corollary~\ref{cor:ii}.b), we conclude that $\widetilde{\cal F} = {\it cl}({\G_{\it id}})$ is a bisimulation, hence $\G \subseteq {\it cl}({\G_{\it id}}) \subseteq {\sim_\Gf}$.
}

\item[b)] $\Gf \not = \Bl$.
Based on the reduction rules implemented in {\CIRC},
it is quite easy to see that the initial set of goals $\G$ is a
$\vdashInd$-consequence of $\behspec \cup \efr{\mathcal{F}}$.
In other words, $\G \subseteq \cl{{\mathcal{F}}_{\it id}}$.
So, if we anticipate a bit, we should show that
$\widetilde{\cal F}=\cl{{\mathcal{F}}_{\it id}}$ is a bisimulation,
\textit{i.e.}, according to Corollary~\ref{cor:ii},
$\behspec \cup \efr{\mathcal{F}} \vdashInd \efr{\delta_{\GtrlG}(\mathcal{F})}$.
This is achieved by proving that
$\behspec \cup \efr{\mathcal{F}} \vdashInd \G_i (i \in \overline{0,n})$
(note that $\efr{\delta_{\GtrlG}(\mathcal{F})} \subseteq \bigcup_{i \in \overline{0,n}} \G_i$,
according to \textsf[Derive]).
The proof is by induction on $j$, where $n-j$ is the current
proof step, and by
case analysis on the {\CIRC} reduction rules applied at each step.

We further provide a sketch of the proof.\\
The \emph{base case} $j = n$ follows immediately, as $\behspec \cup \efr{\mathcal{F}} \vdashInd \G_{n} = \emptyset$.\\
For the \emph{induction step} we proceed as follows. Let $\efr{e} \in \G_{j}$. If $\efr{e} \in \G_{{j+1}}$ then $\behspec \cup \efr{\mathcal{F}} \vdashInd \efr{e}$ by the induction hypothesis. If $\efr{e} \not \in \G_{{j + 1}}$ then, for example, if \textsf{[Reduce]} was applied then it holds that $\behspec \cup {\mathcal F_{j}} \vdashInd \efr{e}$. Recall that ${\mathcal F_{j}} \subseteq \efr{\mathcal F}$, so $\behspec \cup \efr{\mathcal F} \vdashInd \efr{e}$ also holds. The result follows in a similar fashion for the application of \textsf{[Derive]} or \textsf{[Simplify]}.
\end{itemize}
\qed
\end{proof}

\begin{remark}
The soundness of the proof system we describe in this paper does not follow directly from Theorem 3 in \cite{rosu-lucanu-2009-calco}.
This is due to the fact that
we do not have an experiment-based definition of bisimilarity.
So, even though the mechanism we use for proving
$\behspec \cup \efr{\mathcal{F}} \vdashInd \efr{\delta_{\GtrlG}(\mathcal{F})}$ (for the case $\Gf \not = \Bl$)
is similar to the one described in
\cite{rosu-lucanu-2009-calco},
the current soundness proof is conceived in terms of bisimulations
(and not experiments).
\end{remark}

\begin{remark}
The entailment relation $\vdashInd$ that {\CIRC} uses
for checking the equivalence of generalized regular expressions
is an instantiation of the parametric entailment
relation $\vdash$ from the proof system in \cite{rosu-lucanu-2009-calco}.
This approach allows {\CIRC} to reason automatically on a large class
of systems which can be modeled as non-deterministic coalgebras.
\end{remark}

As already stated, our final goal is to use {\CIRC} as a
decision procedure for the bisimilarity of generalized regular expressions.
That is, whenever provided a set of expressions, the prover stops
with a yes/no answer w.r.t. their equivalence.
In this context,
an important aspect is that the sub-coalgebra generated
by an expression $\eps \in {\Exp_{\Gf}}$ by repeatedly applying
$\delta_{\Gf}$ is, in general, infinite. Take for example
the non-deterministic functor \ans{C.1.26} $\Rf = \itB \times \id$ associated to infinite streams, and
consider the property
\(
\mu x . \emp \oplus r\langle x \rangle =
\mu x . r\langle x \rangle
\).
In order to prove this, {\CIRC} builds an infinite proof sequence
by repeatedly applying  $\delta_{\Rf}$ as follows:
\begin{center}
\begin{tabular}{rcl}
$\delta_{\Rf}(\mu x . \emp \oplus r\langle x \rangle)$ &
$=$ &
$\delta_{\Rf}(\mu x . r\langle x \rangle)$\\
& $\downarrow$ &\\
$\langle 0, \emp \oplus (\mu x . \emp \oplus r\langle x \rangle)\rangle$ &
$=$ &
$\langle 0, \mu x . r\langle x \rangle \rangle$\\[2ex]

$\delta_{\Rf}(\emp \oplus (\mu x . \emp \oplus r\langle x \rangle))$ &
$=$ &
$\delta_{\Rf}(\mu x . r\langle x \rangle)$\\
& $\downarrow$ &\\
$\langle 0, \emp \oplus \emp \oplus (\mu x . \emp \oplus r\langle x \rangle)\rangle$ &
$=$ &
$\langle 0, \mu x . r\langle x \rangle \rangle$ [\ldots\!]
\end{tabular}
\end{center}
\label{just:ACI}
In this case, the prover would never stop. We observed in Section~\ref{sec:dp} that Theorem~\ref{thm:kleene} guarantees we can associate a finite coalgebra to a certain expression. In the proof of the aforementioned theorem, which is presented in \cite{brs_lmcs}, it is shown 
that the axioms for associativity, commutativity and idempotence (ACI) of $\oplus$
guarantee finiteness of the generated sub-coalgebra (note that these axioms have also been proven
sound w.r.t. bisimulation).
ACI properties can easily be specified in {\CIRC}
as the prover is an extension of Maude, which has a
powerful matching modulo ACUI (ACI plus unity) capability.
The idempotence is given by the equation $\eps \oplus \eps = \eps$, and
the commutativity and associativity are specified as attributes
of $\oplus$. It is interesting to remark that for the powerset functor termination is guaranteed without the axioms, because the coalgebra structure on the expressions for the powerset functor already includes ACI (since $\pow(\Exp)$ is itself a join-semilattice).

\begin{theorem}\label{thm:decProc}
Let $\G$ be a set of proof obligations over generalized regular expressions.
{\CIRC} can be used as a decision procedure for the equivalences in $\G$,
that is, it can \ans{C.1.27} decide whenever a goal $(\E_1, \E_2) \in \G$ is a true or false equality.
\end{theorem}

\begin{proof}
\ans{C.3.31}
Note that as proven in~\cite{brs_lmcs}, the ACI axioms for $\oplus$ guarantee that $\delta_\Gf$ is applied for a finite number of times in the generation of the sub-coalgebra associated to a $\Gf$-expression. Therefore, it straightforwardly follows that by implementing the ACI axioms in {\CIRC} (as attributes of $\oplus$), the set of new goals obtained by applying $\delta_\Gf$ is finite.
In these circumstances, whenever {\CIRC} stops according to the
reduction rule \textsf{[Done]},
the initial proof obligations are bisimilar. On the other hand, whenever it terminates
with \textsf{[Fail]}, the goals are not bisimilar.
\qed
\end{proof}

\section{A {\CIRC}-based Tool}
\label{sec:caseStudy}

We have implemented a tool that, when provided with a functor $\Gf$,
automatically generates a specification for {\CIRC} which can then be used in order to automatically check whether two $\Gf$-expressions are bisimilar. The tool is implemented as a metalanguage application in Maude. It can be downloaded from the address \url{http://goriac.info/tools/functorizer/}. In order to start the tool, one needs to launch Maude along with the extension Full-Maude and load the downloaded file using the command \code{in functorizer.maude .}

The general use case consists in providing the join-semilattices, the alphabets and the expressions. After these steps, the tool automatically checks if the provided expressions are guarded, closed and correctly typed. If this check succeeds, then it outputs a specification that can be further processed by {\CIRC}. In the end, the prover outputs either the bisimulation, if the expressions are equivalent, or a negative answer, otherwise.

We present two case studies in order to emphasize the high degree of generality for the types of systems we can handle, and show how the tool is used. 

\begin{example}
\label{eg:mealy}
We consider the case of Mealy machines, which are coalgebras
for the functor $(\itB \times \id)^A$. 

Formally, a Mealy machine is a pair $(S,\alpha)$ consisting of a set
$S$ of states and a transition function $\alpha\colon S\to (\itB \times
S)^A$, which for each state $s\in S$ and input $a\in A$ associates an
output value $b$ and a next state $s'$. Typically, we write
$\alpha(s)(a) = (b,s') \Leftrightarrow
\xymatrix{*+[o][F]{s}\ar[r]^{a|b} &*+[o][F]{s'}}$.

\medskip
In this example and in what follows we will consider for the output the two-value join-semilatice  $\itB = \{0,1\}$ (with $\bottom_\Bl = 0$) and for the input alphabet $A = \{a,b\}$. 
The expressions for Mealy machines are given by the grammar:
\[
\begin{array}{rl}
E_{\phantom{0}} &::\!= \emp \mid x \mid E \oplus E \mid \mu x . E_{2} \mid a(r<E>) \mid b(r<E>) \mid a(l<E_1>) \mid b(l<E_1>) \\
E_{1} &::\!= \emp \mid E_1 \oplus E_1 \mid 0 \mid 1\\
E_{2} &::\!= \emp \mid E_2 \oplus E_2 \mid \mu x . E_{2} \mid a(r<E>) \mid b(r<E>) \mid a(l<E_1>) \mid b(l<E_1) \\
\end{array}
\]

Intuitively, an expression of shape $a(l<E_1>)$ specifies a state that for an input $a$ has an output value specified by $E_1$. For example, the expression $a(l<1>)$ specifies a state that for input $a$ outputs $1$, whereas in the case of $a(l<\emp>)$ the output is $0$. An expression of shape $a(r<E>)$ specifies a state that for a certain input $a$ has a transition to a new state represented by $E$. For example, the expression $\mu x. a(r<x>)$ states that for input $a$, the machine will perform a ``$a$-loop" transition, whereas $a(r<\emp>)$ states that for input $a$ there is a transition to the state denoted by $\emp$. It is interesting to note that a state will only be fully specified in what concerns transitions and output (for a given input $a$ if both $a(l<E_1>)$ and $a(r<E>)$ appear in the expression (combined by $\oplus$). In the case only transition (resp. output) are specified, the underspecification is solved by setting the target state (resp. output) to $\emp$ (resp. $\bot_B = 0$). 
\end{example}

Next, to provide the reader with intuition, we will explain how one can reason on the bisimilarity of two simple expressions, by constructing bisimulation relations. Later on, we show how {\CIRC} can be used in conjunction with our tool in order to act as a decision procedure when checking equivalence of two expressions, in a fully automated manner.

%
%


We will start with the expressions $\eps_1 = \mu x . a(r<x>)$ and $\eps_2 = \emp$.We have to build a bisimulation relation  $\R$ on $\Gf$-expressions, such that  $(\eps_1, \eps_2) \in \R$. We do this in the following way: we start by taking $\R=\{(\E_1,\E_2)\}$ and we check whether this is already a bisimulation, by considering the output values and transitions and check whether no new expressions appear in this process. If new pairs of expressions appear we add them to $\R$ and repeat the process. Intuitively, this can be represented as follows:

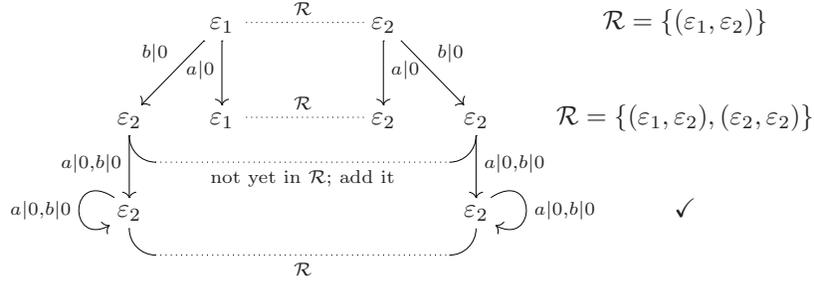
\begin{figure}[h]
\centering
$\xymatrix@C=0.7cm@R=0.7cm{
& {\eps_1}\ar[d]_{a|0}\ar@{..}[rr]^\R\ar[dl]_{b|0} & & {\eps_2}\ar[d]^{a|0}\ar[dr]^{b|0} & & \R = \{(\eps_1, \eps_2)\} \\
{\eps_2 }\ar[d]_{a|0,b|0} & {\eps_1}\ar@{..}[rr]^\R & & {\eps_2} & {\eps_2}\ar@{..} `d_l[llll] `[llll]^{\text{not yet in }\R \text{; add it}} [llll] \ar[d]^{a|0,b|0}& \R = \{(\eps_1, \eps_2), (\eps_2, \eps_2)\}\\
{\eps_2}\ar@(ul,dl)_{a|0,b|0}& & & & {\eps_2}\ar@(ur,dr)^{a|0,b|0} \ar@{..} `d_l[llll] `[llll]^\R [llll] & \checkmark
}$
\caption{Bisimulation construction}
\label{fig:mealy-a-noTrans-BC}
\end{figure}

In the figure above, and as before, we use the notation $\xymatrix@C=0.7cm@R=0.7cm{
 {\eps_1}\ar@{-}[r]^\R&\eps_2}$ to denote $(\eps_1,\eps_2)\in\R$.
As illustrated in Figure~\ref{fig:mealy-a-noTrans-BC}, $\R = \{(\eps_1, \eps_2), (\eps_2, \eps_2)\}$ is closed under transitions and is therefore a bisimulation. Hence, $\eps_1 \sim_\Gf \eps_2$.

The proved equality $\emp = \mu x . a(r< x >)$ might seem unexpected, if the reader is familiar with labelled transition systems. The equality is sound because these are expressions specifying behavior of a Mealy machine and, semantically, both denote the function that for every non-emtpy word outputs $0$ (the semantics of Mealy machines is given by functions $B^{A^+}$, intuitively one can think of these expressions as both denoting the empty language). This is visible if one draws the automata corresponding to both expressions (say, for simplicity, the alphabet is $A=\{a\}$):
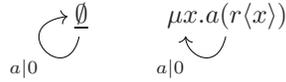
\[
\xymatrix{\emp\ar@(d,l)^{a|0} &  \mu x . a(r< x >) \ar@(d,l)^{a|0}}
\]
Note that (i) the $\emp$ expression for Mealy machines is mapped with $\delta$ to a function that for input $a$ gives $<0, \emp>$, which represents a state with an $a$-loop to itself and output $0$; (ii) the second expression specifies explicitly an $a$-loop to itself and it also has output $0$, since no output value is explicitly defined.  
Now, also note that similar expressions for labelled transition systems (LTS), or coalgebras of the functor $\pow(-)^A$, would not be bisimilar since one would have an a-transition and the other one not. This is because the $\emp$ expression for LTS really denotes a deadlock state. In operational terms they would be converted to the systems
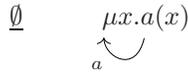
\[
\xymatrix{\emp &  \mu x . a( x ) \ar@(d,l)^a}
\] 
which now have an obvious difference in behavior. 

By performing a similar reasoning as in the example above one can show that the expressions $\eps_1 = \mu x . a(r<x>) \oplus b(r<x>)$ and $\eps_2 = \mu x . a(r<x>)$ are bisimilar, and the bisimulation relation is built as illustrated in Figure~\ref{fig:mealy-ab-trans-BC}:

\begin{figure}[H]
\centering
$\xymatrix@C=0.7cm@R=0.7cm{
& {\eps_1}\ar[d]_{a|0}\ar@{..}[rr]^\R\ar[dl]_{b|0} & & {\eps_2}\ar[d]^{a|0}\ar[dr]^{b|0} & & \R = \{(\eps_1, \eps_2)\} \\
{\eps_1 }\ar[d]_{a|0,b|0} & {\eps_1}\ar@{..}[rr]^\R & & {\eps_2} & {\emp}\ar@{..} `d_l[llll] `[llll]^{\text{not yet in }\R \text{; add it}} [llll] \ar[d]^{a|0,b|0}& \R = \{(\eps_1, \eps_2), (\eps_1, \emp)\}\\
{\eps_1}\ar@(ul,dl)_{a|0,b|0}& & & & {\emp}\ar@(ur,dr)^{a|0,b|0} \ar@{..} `d_l[llll] `[llll]^\R [llll] & \checkmark
}$
\caption{Bisimulation construction}
\label{fig:mealy-ab-trans-BC}
\end{figure}
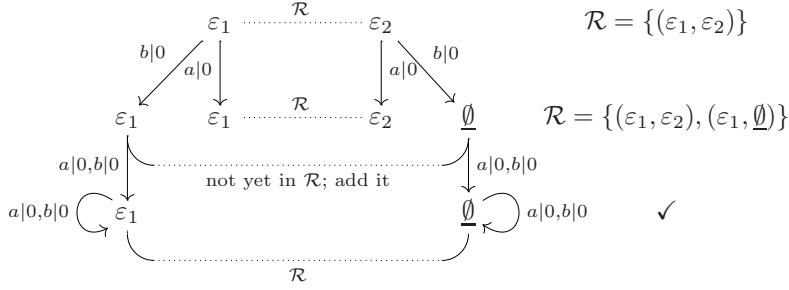

Let us further consider the Mealy machine depicted in Figure~\ref{fig:mealy1}, where all states are bisimilar.

\begin{figure}[H]
\centering
$\xymatrix@C=1.9cm@R=0.5cm{
*+[o][F]{s_1}\ar@/^/[r]_{a|0}\ar@/_/[d]^{b|1} &
*+[o][F]{\phantom{s_3}}\ar@(dr,ur)_{a|0}\ar@/^/[d]_{b|1} \\
*+[o][F]{\phantom{s_4}}\ar@(dr,ur)_{b|1}\ar@(dl,ul)^{a|0} &
*+[o][F]{s_2}\ar@(dr,ur)_{b|1}\ar@(dl,ul)^{a|0}
}$
\caption{Mealy machine: $s_1 \sim s_2$}
\label{fig:mealy1}
\end{figure}
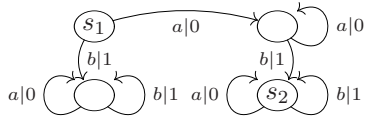

We show how to check the equivalence of two expression characterizing the states $s_1$ and $s_2$, in a fully automated manner, using {\CIRC}. These expressions are
\ans{C.1.30}
$\eps_1 = \mu x . b(l<1>) \oplus b(r<\eps_2>) \oplus a(\mu y . a(r<y>) \oplus b(r<\eps_2>) \oplus b(l<1>))$ and $\eps_2 = \mu x . b(l<1>) \oplus b(r<x>) \oplus a(r<x>)$, respectively.

In order to check bisimilarity of $\eps_1$ and $\eps_2$ we
load the tool and define the semilattice $\Bl = \{0,1\}$ and the alphabet
${\it A} = \{a, b\}$:
\\[0.5ex]
\code{
(jslt B is 0 1
  bottom 0 .
  0 v 0 = 0 .
  0 v 1 = 1 .
  1 v 1 = 1 .
endjslt)}\\
\code{(alph A is a b endalph)}

\smallskip

We provide the functor $\Gf$ using the command \code{(functor (B x Id)}\verb#^#\code{A .)}.
The command \code{(set goal ... .)} specifies the goal we want to prove:
\begin{alltt}
\fontsize{9}{10}
\selectfont(set goal
 \verb#\#mu X:FixpVar . b(l<1>) (+) a(l<0>) (+) b(r<X:FixpVar>) (+)
                 a(r<X:FixpVar>) =
 \verb#\#mu X:FixpVar . b(l<1>) (+) b(<\verb#\#mu X:FixpVar . b(l<1>) (+)
                 b(r<X:FixpVar>) (+) a(r<X:FixpVar>)>) (+)
                 a(\verb#\#mu Y:FixpVar . a(r<Y:FixpVar>) (+)
                 b(<\verb#\#mu X:FixpVar . b(l<1>) (+) a(l<0>) (+)
                 b(r<X:FixpVar>) (+) a(r<X:FixpVar>)>) (+) b(l<1>)) .)
\end{alltt}

\smallskip
In order to generate the {\CIRC} specification
we use the command \code{(generate coalgebra .)}.
Next we need to load {\CIRC} along with the resulting specification
and start the proof engine using the command
\code{(coinduction .)}.

As already shown, behind the scenes, {\CIRC} builds a bisimulation relation
that includes the initial goal.
The proof succeeds and the output consists of (a subset of) this bisimulation:
\begin{alltt}
\fontsize{9}{10}\selectfont{}Proof succeeded.
  Number of derived goals: 2
  Number of proving steps performed: 50
  Maximum number of proving steps is set to: 256

Proved properties:
- phi (+) (\verb#\#mu X . a(l<0>) (+) a(r<X>) (+) b(l<1>) (+) b(r<X>)) =
  phi (+) (\verb#\#mu Y . a(r<Y>) (+) b(l<1>) (+)
  b(r<\verb#\#mu X . a(l<0>) (+) a(r<X>) (+) b(l<1>)(+)b(r<X>)>))

- \verb#\#mu X . a(l<0>) (+) a(r<X>) (+) b(l<1>) (+) b(r<X>) =
  \verb#\#mu Z . a(r<\verb#\#mu Y . a(r<Y>) (+) b(l<1>) (+)
          b(r<\verb#\#mu X . a(l<0>) (+) a(r<X>) (+) b(l<1>) (+) b(r<X>)>)>) (+)
          b(l<1>) (+) b(r<\verb#\#mu X . a(l<0>) (+) a(r<X>) (+)
          b(l<1>) (+) b(r<X>)>)
 \end{alltt}

For the ease of understanding, here we printed a readable version of the proved properties. In Section~\ref{sec:code}, however, we show that internally each expression is brought to a canonical form by renaming the variables.
Moreover, note that in our tool, $\emp$ is represented by the constant $\code{phi}$. All the examples provided in the current section make use of this convention.

As previously mentioned, {\CIRC} is also able to detect when two
expressions are not equivalent. Take, for instance, the expressions
$\mu x . a(l<0>) \oplus a(r<a(l<1>) \oplus a(r<x>)>)$ and $a(l<0>) \oplus a(r<a(r< \mu x . a(r<x>) \oplus a(l<0>)>) \oplus a(l<1>)>)$, characterizing the states $s_1$ and $s_3$ from the Mealy machines in Fig.~\ref{fig:mealy2}. After following some steps similar to the ones previously enumerated, the proof fails and the output message is
\code{Visible goal [...] failed during coinduction}.

\begin{figure}[H]
\centering
$\xymatrix@C=1cm{
*+[o][F]{s_1}\ar@/^1pc/[r]^{a|0} &
*+[o][F]{s_2}\ar@/^1pc/[l]_{a|1} &
*+[o][F]{s_3}\ar@/^/[r]^{a|0} &
*+[o][F]{s_4}\ar@/^/[r]^{a|1} &
*+[o][F]{s_5}\ar@(dr,ur)_{a|0} & \\
}$
\caption{Mealy machines: $s_1 \not\sim s_3$}
\label{fig:mealy2}
\end{figure}
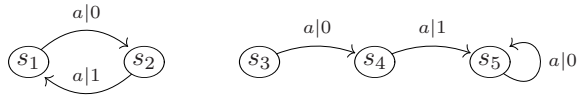

\begin{example}
\label{eg:ccs}


Let us show how one may check strong bisimilarity of
two nondeterministic processes of a non-trivial CCS-like language
with termination, deadlock, and divergence, as studied in \cite{Aceto:1992:TDD:147508.147527}. A process is a guarded, closed term defined by the following grammar:

\newcommand{\success}{\checkmark}
\newcommand{\empsum}{\textnormal{$\Omega$}}

\begin{eqnarray}\label{eq:ccs-gram}
P &::\!=& \success \mid \delta \mid \empsum \mid  a.P \mid P + P \mid x \mid \mu x . P
\end{eqnarray}

where:

\begin{itemize}
\item $\success$ is the constant for successful termination,
\item $\delta$ denotes deadlock,
\item $\empsum$ is the divergent computation (\emph{i.e.}, the undefined process),
\item $a.P$\, is the process executing the action $a$ and then continuing as the process $P$, for any action $a$ from a given set $A$,
\item $P_1 + P_2$ is the non-deterministic process behaving as either $P_1$ or $P_2$, and
\item $\mu x . P$ is the recursive process $P[\mu x . P / x]$.
\end{itemize}

In~\cite{brs_lmcs} is is shown that, up to strong bisimilarity, the above syntax of processes is equivalent to the canonical set of (guarded, closed) regular expressions derived for the functor $1 \myplus {\cal P}_{\omega}(\id)^A$,

\newcommand{\unu}{1}

\[
\begin{array}{lcl}
E   & ::= & \emptyset \mid E \oplus E \mid x \mid \mu x.E \mid l[E_1] \mid r[E_2]\\
E_1 & ::= & \emptyset \mid  E_1 \oplus E_1 \mid \unu \\
E_2 & ::= & \emptyset \mid  E_2 \oplus E_2 \mid a(E_3)\\
E_3 & ::= & \emptyset \mid  E_3 \oplus E_3 \mid \{ E\}
\end{array}
\]

The translation map $(-)^\dagger$ from processes to expressions is defined by
induction on the structure of the process:
\[
\begin{array}{lcl@{\hspace{1.3cm}}lcl}
(\success)^\dagger &=& l[\unu] & (a.P)^\dagger &=& r[a(\{P^\dagger\})]\\
(\delta)^\dagger &=& r[\emptyset] & (P_1+P_2)^\dagger &=&  (P_1)^\dagger \oplus (P_2)^\dagger\\
(\empsum)^\dagger &=& \emptyset & (\mu x. P)^\dagger &=& \mu x. P^\dagger\\
x^\dagger  &=& x \,.
\end{array}
\]

Consider now two processes $P$ and $Q$ over the alphabet $A = \{ a, b \}$:
\[
\begin{array}{lcl}
P & = & \mu x. (a.x + a.P_1 + b.b.\success + b.(\delta + \empsum))\\
Q & = & \mu z. (a.z + b.(\delta + b.\success) + b.\delta)
\end{array}
\]
where $P_1 = \mu y. (a.(y+\delta)+b.\delta+b.(\delta+b.\success)+\delta)$.
Graphically, the two processes can be represented by the following labelled transition systems (for simplicity we omit annotating states with information regarding the satisfiability of successful termination, divergence, and deadlock):

\begin{figure}[H]
\centering
$\xymatrix@C=.7cm@R=.7cm{
&
*+[o][F]{P}\ar@(dl,ul)_{a}\ar[r]^{b}\ar[d]_{a}\ar[dr]^{b} &
*+[o][F]{\phantom{A} } &
*+[o][F]{Q}\ar@(dr,ur)^{a}\ar[d]_{b}\ar[dr]^{b}\\
*+[o][F]{\phantom{A}} &
*+[o][F]{P_1}\ar@(dr,ur)^{a}\ar[l]_{b}\ar[dl]_{b} &
*+[o][F]{\phantom{A}}\ar[d]^{b} &
*+[o][F]{\phantom{A}}\ar[d]_{b} &
*+[o][F]{\phantom{A}}  \\
*+[o][F]{\phantom{A}}\ar[r]^{b} &
*+[o][F]{\phantom{A}} &
*+[o][F]{\phantom{A}} &
*+[o][F]{\phantom{A}} \\
}$
\caption{Nondeterministic processes: $Q \sim P$}
\label{fig:pa}
\end{figure}
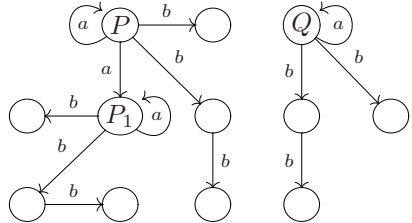

We want to check if the process $P$ is strongly bisimilar to the process $Q$. By using the above translation, process $P$ is represented by the expression
\[
\begin{array}{l@{}l@{}l}
\mu x .
(&r[a(\{\mu y. (&r[a(\{ y \oplus r[\emptyset]\})] \oplus
                 r[b(\{r[\emptyset]\})] \oplus\\
 &              &r[b(\{r[\emptyset] \oplus r[b(\{l[\unu]\})] \})]\oplus
                r[\emptyset]
               )
   \})] \oplus \\
 & \multicolumn{2}{l}{ r[a(\{x\})] \oplus
   r[b(\{r[b(\{ l[\unu]\})] \})] \oplus
   r[b(\{r[\emptyset] \oplus \emptyset\})]
)}\\[1ex]
\end{array}
\]
whereas process $Q$ is represented by the expression
\[
\begin{array}{l@{}l@{}l}
\mu z .
 (&r[a(\{z\})] \oplus
        r[b(\{r[\emptyset] \oplus r[b(\{l[\unu]\})] \})] \oplus
        r[b(\{r[\emptyset]\})]
       ).
\end{array}
\]


In order to use the tool, one needs to specify the semilattice, 
the alphabet, the functor, and the goal in a manner similar to the one previously presented:
\\[0.5ex]
\code{
(jslt B is 1
  bottom 1 .
  1 v 1 = 1 .
endjslt)}\\
\code{(alph A is a b endalph)}\\
\code{(functor B + (P Id)}\verb#^#\code{A .)}
\begin{alltt}
\fontsize{9}{10}\selectfont(set goal \verb#\#mu X:FixpVar .
          r[ a( \{ X:FixpVar \} ) ] (+)
          r[ a( \{ \verb#\#mu Y:FixpVar .
                  r[ a( \{ Y:FixpVar (+) r[ phi ] \} ) ] (+)
                  r[ b( \{ r[ phi ] \} ) ] (+)
                  r[ b( \{ r[ phi ] (+) r[ b( \{ l[ 1 ] \} ) ] \} ) ] (+)
                  r[ phi ]
              \} )
          ] (+)
          r[ b( \{ r[ b( \{ l[ 1 ] \} ) ] \} ) ] (+)
          r[ b( \{ r[ phi ] (+) phi \} ) ]
          =
          \verb#\#mu Z:FixpVar .
          r[ a( \{ Z:FixpVar \} ) ] (+)
          r[ b( \{ r[ phi ] (+) r[ b( \{ l[ 1 ] \} ) ] \} ) ] (+)
          r[ b( \{ r[ phi ] \} ) ]  .)
\end{alltt}

For the generated specification {\CIRC} terminates and outputs a positive result:

\begin{alltt}
\fontsize{9}{10}\selectfont{}Proof succeeded.
  Number of derived goals: 15
  Number of proving steps performed: 58
  Maximum number of proving steps is set to: 256

Proved properties:
- r[phi] (+) (\verb#\#mu Y. r[phi] (+) r[a(\{r[phi] (+) Y\})] (+) r[b(\{r[phi]\})]
  (+) r[b(\{r[phi] (+) r[b(\{l[1]\})]\})])
  =
  \verb#\#mu Z. r[a(\{Z\})] (+) r[b(\{r[phi]\})] (+) r[b(\{r[phi] (+) r[b(\{l[1]\})]\})]
- r[b(\{l[1]\})] = r[phi] (+) r[b(\{l[1]\})]
- \verb#\#mu Y. r[phi] (+) r[a(\{r[phi] (+) Y\})] (+) r[b(\{r[phi]\})] (+)
  r[b(\{r[phi] (+) r[b(\{l[1]\})]\})]
  =
  \verb#\#mu Z. r[a(\{Z\})] (+) r[b(\{r[phi]\})] (+) r[b(\{r[phi] (+) r[b(\{l[1]\})]\})]
- \verb#\#mu X. r[a(\{X\})] (+) r[a(\{\verb#\#mu Y. r[phi] (+) r[a(\{r[phi] (+) Y\})] (+)
  r[b(\{r[phi]\})] (+) r[b(\{r[phi] (+) r[b(\{l[1]\})]\})]\})] (+)
  r[b(\{r[phi] + phi\})] (+) r[b(\{r[b(\{l[1]\})]\})]
  =
  \verb#\#mu Z. r[a(\{Z\})] (+) r[b(\{r[phi]\})] (+) r[b(\{r[phi] (+) r[b(\{l[1]\})]\})]
\end{alltt}

\end{example}

\subsection{Implementation}
\label{sec:code}

In this section we present details on the implementation of the algebraic specification given in Section~\ref{sec:algSpec}, based on the examples from Section~\ref{sec:caseStudy}.

In order to generate the algebraic specifications for {\CIRC} when provided a functor and two expressions we used the Maude system \cite{DBLP:conf/maude/2007}. We choose it for its suitability for performing equational and rewriting logic based computations, and because of its reflective properties allowing for the development of advanced metalanguage applications. As the technical aspects on how to work at the meta-level are beyond the scope of this paper, we refrain from presenting them and show, instead, what the generated specifications consist of.

Most of the algebraic specifications from Section~\ref{sec:algSpec} have a straightforward implementation in Maude. Consider, for instance, the case of Mealy machines presented in Example~\ref{eg:mealy}. The generated grammars for functors (\ref{eq:fun-gram}) and expressions (Definition~\ref{def:expr}) are coded as:
\begin{alltt}
\fontsize{9}{10}\selectfont{}sort Functor .                          sorts Exp ExpStruct Alph Slt .
sorts AlphName SltName .                subsort Exp < ExpStruct .
subsort SltName < Functor .             enum A is a b . enum B is 0 1 .
                                        subsort A < Alph .
op A : -> AlphName .                    subsort B < Slt .
op B : -> SltName .                     
op G : -> Functor .                     op _`(+`)_ : Exp Exp -> Exp .
op Id : -> Functor .                    op _`(_`) : Alph Exp -> Exp .
op _+_ : Functor Functor -> Functor .   op \verb#\#mu_._ : FixpVar Exp -> Exp .
op _^_ : Functor AlphName -> Functor .  ops l<_> r<_> : Exp -> Exp .
op _x_ : Functor Functor -> Functor .   op phi : -> Exp .

                          eq G = (B x Id) ^ A .
\end{alltt}

Most of the syntactical constructs are Maude-specific: \code{sorts} and \code{subsort} declare the sorts we work with and, respectively, the relations between them; \code{op} declares operators; \code{eq} declares equations (the equation in our case defines the shape of the functor \code{G}). The only {\CIRC}-specific construct, \code{enum}, is syntactic sugar for declaring enumerable sorts, \emph{i.e.}, sorts that consist only of the specified constants. As a side note, if brackets (\code{(}, \code{[}, \code{\{})  are used in the declaration of an operation, then they must be preceded by a backquote (\code{`}).

As mentioned in Section~\ref{sec:prelim}, in order to guarantee the finiteness of our procedure, one needs to include the ACI axioms for \code{(+)}. \ans{C.1.5} Moreover, we have observed that the unity axiom for \code{(+)} plays an important role in decreasing the number of states generated by the repeated application of $\delta_{\Gf}$, therefore improving the overall time performance of the tool.
\ans{C.2.23} For example, the number of rewritings {\CIRC} performed in order to prove the bisimilarity of $\E_1$ and $\E_2$ in Figure~\ref{fig:mealy-ab-trans-BC} was halved when the unity axiom was used.

By turning on the axiomatization flag using the command \code{(axioms on .)}, the following code is generated:

\begin{alltt}
\fontsize{9}{10}\selectfont{}op _`(+`)_ : Exp Exp -> Exp [assoc comm] .
eq E:Exp (+) E:Exp = E:Exp .
eq E:Exp (+) phi = E:Exp .
\end{alltt}

It is an obvious question why not to add other axioms to the tool, since the unity axiom has improved performance. At this stage we do not have studied in detail how much adding other axioms would help. It is in any case a trade-off on how many extra axioms one should include, which will get the automaton produced from an expression closer to the minimal automaton, and how much time the tool will take to reduce the expressions in each step modulo the axioms. For classical regular expressions, there is an interesting empirical study on this~\cite{derivatives-jfp09}. We leave it as future work to carry on a similar study for our expressions and axioms.

The process of substituting fixed-point variables has a natural implementation. We present the equations handling the basic expressions $\emp$ and $x$, and the operation \code{(+)}:
\begin{alltt}
\fontsize{9}{10}\selectfont{}op _`[_/_`] : Exp Exp FixpVar -> Exp .
eq phi [ E:Exp / X:FixpVar ] = phi .
ceq Y:FixpVar [ E:Exp / X:FixpVar ] = E:Exp if (X:FixpVar == Y:FixpVar) .
eq Y:FixpVar [ E:Exp / X:FixpVar ] = Y:FixpVar [owise] .
eq (E1:Exp (+) E2:Exp) [ E:Exp / X:FixpVar ] = 
   (E1:Exp [E:Exp / X:FixpVar]) (+) (E2:Exp [E:Exp / X:FixpVar]) .
\end{alltt}
  
In order to avoid matching problems and to overpass the fact that in Maude one cannot handle an equation that has fresh variables in its right-hand-side (\emph{i.e.}, they do not appear in the left-hand-side), we replace expression variables with parameterized constants: \code{op var : Nat -> FixpVar .} The operation that obtains this canonical form has an inductive definition on the structure of the given expression and makes use of the substitution operation presented above. For this reason, the bisimulation {\CIRC} builds contains parameterized constants instead of the user declared variables. The property proved in Example~\ref{eg:ccs} is, therefore, written as:

\begin{alltt}
\fontsize{9}{10}\selectfont{}\verb#\#mu var(2) . r[a(\{var(2)\})] (+) r[a(\{\verb#\#mu var(1) . r[phi] (+)
r[a(\{r[phi] (+) var(1)\})] (+) r[b(\{r[phi]\})] (+) r[b(\{r[phi] (+)
r[b(\{l[1]\})]\})]\})] (+) r[b(\{r[phi] (+) phi\})] (+) r[b(\{r[b(\{l[1]\})]\})]
=
\verb#\#mu var(1) . r[a(\{var(1)\})] (+) r[b(\{r[phi]\})] (+)
r[b(\{r[phi] (+) r[b(\{l[1]\})]\})]
\end{alltt}

The most important part of the algebraic specification consists of the equations defining the operations $\it \delta\_(\_)$, $\it Plus\_(\_,\_)$, and $\it Empty$. Most of these equations are implemented as presented in \cite{brs_lmcs}. The only difficulties we encountered were for the exponentiation case, as Maude does not handle higher-order functions. Without entering into details, as a workaround, we introduced a new sort \code{Function < ExpStruct} and an operation \code{{\bf \textbackslash}. : ExpoCase Alph Functor ExpStruct -> Function} in order to emulate function-passing. The first argument is used to memorize the origin where the exponentiation ingredient is encountered: $\delta$, $\it Plus$, or $\it Empty$. Its purpose is purely technical -- we use it in order to avoid some internal matching problems. The other three parameters are those of the structured expression $\lambda {.}(a,\FtrlG,\sigma)$ presented in Section~\ref{sec:algSpec}: a letter in the alphabet, an ingredient, and some other structured expression.

Another thing worth describing is the way we enable {\CIRC} to prove equivalences when the powerset functor occurs. Namely, we present how interpolant (\ref{srl:pow}) is implemented. Recall that we want to show that two sets of expressions are equivalent, which means that for each expression in the first set there must be an equivalent one in the second set and vice-versa.

\newcommand{\largeusq}{{\fontsize{12}{13}\selectfont\_}}
In order to handle sets of structured expressions we introduce a new sort, \code{ExpStructSet} as a supersort for \code{ExpStruct}. We also consider the set separator \code{\largeusq,\largeusq \,\,: ExpStructSet ExpStructSet -> ExpStructSet [assoc,comm]}, the empty set \code{emptyS : -> ExpStructSet}, and the set wrapping operation \code{\{\largeusq\} : ExpStructSet -> ExpStruct}. In order to mimic universal quantification over a set, we use a special constant referred to as token ``\code{[/]}''. In what follows, we consider two variables of sort \code{ExpStructSet}: \code{ES} and \code{ES'}, and two variables of sort \code{ExpStructSet}: \code{ESS} and \code{ESS'}. We now describe the process of finding the equivalence between two sets:
\begin{itemize}
\item whenever encountering two wrapped expression sets we add the universal quantification token to each of them in two distinct goals:
\begin{alltt}
\fontsize{9}{10}\selectfont srl \{ESS\} = \{ESS'\} => \{[/] ESS\} = \{ESS'\} /\verb#\# \{ESS\} = \{[/] ESS'\} .
\end{alltt}
\item iterate through the expressions on the left-hand-side (similarly for the other direction):
\begin{alltt}
\fontsize{9}{10}\selectfont srl \{[/] (ES , ESS)\} = \{ESS'\} =>
     \{[/] ES\} = \{ESS'\} /\verb#\# \{[/] ESS\} = \{ESS'\} .
 srl \{ESS\} = \{[/] (ES' , ESS')\} =>
     \{ESS\} = \{[/] ES'\} /\verb#\# \{ESS\} = \{[/] ESS'\} .
\end{alltt}
\item when left with one expression on the left-hand-side, start iterating through the expressions on the right-hand-side until finding an equivalence (similarly for the other direction):
\begin{alltt}
\fontsize{9}{10}\selectfont srl \{[/] ES\} = \{ES' , ESS'\} => ES = ES' \verb#\#/ \{[/] ES\} = \{ESS'\} .
 srl \{ES , ESS\} = \{[/] ES'\}  => ES = ES' \verb#\#/ \{ESS\} = \{[/] ES'\} .
\end{alltt}
\item if no equivalence has been found, transform the current goal into a visible failure:
\begin{alltt}
\fontsize{9}{10}\selectfont srl \{ESS\} = emptyS => true = false .
 srl emptyS = \{ESS\} => true = false .
\end{alltt}

\end{itemize}

Finally, the type checker for structured expressions has a straightforward implementation. Its code does not appear in the generated specification as it is only used when the tool receives the expressions as input. This prevents obtaining the specification and starting the prover in case invalid expressions are provided.

\section{Discussion}
\label{sec:concl}
One of the major contributions of this paper is that we provided a decision procedure for the bisimilarity of generalized regular expressions. In order to enable the implementation of the decision procedure, we have exploited an encoding of coalgebra into algebra, and we formalized the equivalence between the coalgebraic concepts associated to non-deterministic coalgebras \cite{brs_lmcs} and their algebraic correspondents. This led to the definition of algebraic specifications (\algspec) that model both the language and the coalgebraic structure of expressions.
Moreover, we defined an equational deduction relation ($\vdashInd$), used on the algebraic side for reasoning on the bisimilarity of expressions.

The most important result of the parallel between the coalgebraic and algebraic approaches is given in Corollary~\ref{cor:ii}, which formalizes the definition of the bisimulation relations in algebraic terms. Actually, this result is the key for proving the soundness of the decision procedure implemented in the automated prover {\CIRC} \cite{lucanu-etal-2009-calco}. As a coinductive prover, {\CIRC} builds a relation $\cal F$ closed under the application of $\delta_\itG$ with respect to $\vdashInd$ ($\algspec \cup \efr{\cal F} \vdashInd \efr{\delta_\itG({\cal F})}$), hence automatically computing a bisimulation the initial proof obligations belong to.

The approach we present in this paper enables {\CIRC} to perform
reasoning based on bisimulations (instead of experiments
\cite{rosu-lucanu-2009-calco}). This way, the prover is extended to
checking bisimilarity in a large class of systems that can be
modeled as non-deterministic coalgebras. Note that the constructions above are all automated -- the (non-trivial) {\CIRC} algebraic specification describing $\algspec$, together with the interpolants implementing $\vdashInd$ are generated with the Maude tool presented in Section~\ref{sec:caseStudy}.

We now mention some of the existing coalgebraic based tools for proving bisimilarity and the main differences with the tool presented in this paper.  CoCasl~\cite{DBLP:conf/fase/HausmannMS05} and CCSL~\cite{ccsl} are tools that can generate proof obligations for 
theorem provers from coalgebraic specifications.  In~\cite{DBLP:conf/fase/HausmannMS05} several tactics for interactive and 
automatic bisimulation building are implemented in Isabelle/HOL and are used 
to derive bisimilarities for translated specifications from CoCasl. The main difference between our tool and CoCasl or CCSL is that, given a functor, the tool derives a  specification language for which equivalence is decidable (that is, it is automatic and not interactive). CIRC~\cite{goguen-lin-rosu-2000-ase,rosu-lucanu-2009-calco}, on top of which the current tool is built, is based on hidden logic~\cite{rosu-thesis} and uses a partial decision procedure for proving bisimilarities via implicit construction of bisimulations. Our tool can be seen as an extension of CIRC to a fully automatic theorem prover for the class of non-deterministic coalgebras. We stress the fact that the focus of this paper was on a language for which equivalence is decidable. Tools such as CoCasl, CCSL or {\CIRC} have a more expressive language, where one can, for instance, specify streams which in our language could not be specified (intuitively, the streams we can specify in our language are eventually periodic). In those tools decidability of equivalence can however not be guaranteed.

There are several directions for future work. 

Extending the class of systems to include quantitative coalgebras (such as weighted automata and Markov chains) will  enlarge the scope of applicability of the tool. The challenge in this extension arises from the fact that the definition of expressions for quantitative coalgebras involving the distribution monad is not as modular as for the other functors (for details see~\cite{bbrs_ic}). This is a consequence of the fact that the sum of two valid expressions might not be a valid expression anymore (since in distributions we require that the sum of probabilities add up to $1$).  Moreover, calculating bisimulation relations in the quantitative setting will encompass metric manipulation, which is currently not implemented in \CIRC.

To improve usability, building a graphical interface for the tool is an obvious next step. The graphical interface should ideally allow the specification of expressions by means of  systems of equations (which are then solved internally) or even by means of an automaton, which would then be translated to an expression using Kleene's theorem. We also would like to explore how adding more axioms than ACI to the prover (that is, each step of the bisimulation checking is performed modulo more equations) improves the performance. Our experience so far shows that by adding the axiom for the distribution of the $\emp$ expression through the constructors, \textit{i.e.} $\emp \oplus \eps = \eps$, the prover works significantly faster.

We have not yet studied complexity bounds for the algorithms presented in this paper. We conjecture however that the bounds will be very similar to the already known for classical regular expressions~\cite{K08a,worthington}. Further explorations in this direction are left as future work.

\paragraph{Acknowledgments} We would like to thank the referees for
the many constructive comments, which greatly helped us to improve the
paper. The authors are also grateful for useful comments from Luca Aceto, Filippo Bonchi, and Miguel Palomino Tarjuelo. The work of Georgiana Caltais and Eugen-Ioan Goriac has been partially supported by the project `Meta-theory of Algebraic Process Theories' (nr.~100014021) of the Icelandic Research Fund. The work of Eugen-Ioan Goriac has also been partially supported by the project `Extending and Axiomatizing Structural Operational Semantics: Theory and Tools' (nr. 110294-0061) of the Icelandic Research Fund.
The work of Dorel Lucanu has been partially supported by the PNII grant DAK project Contract 161/15.06.2010, SMIS-CSNR 602-12516. The work of Alexandra Silva was partially funded by ERDF - European
Regional Development Fund through the COMPETE Programme and by
Funda‹o para a Cincia e a Tecnologia, Portugal within projects {\tt
FCOMP-01-0124-FEDER-020537} and {\tt SFRH/BPD/71956/2010}.

\bigskip 
\noindent\textbf{References}

\begin{thebibliography}{10}

\bibitem{Aceto:1992:TDD:147508.147527}
L.~Aceto and M.~Hennessy.
\newblock Termination, deadlock, and divergence.
\newblock {\em J. ACM}, 39:147--187, January 1992.

\bibitem{sbmf}
M.~Bonsangue, G.~Caltais, E.-I. Goriac, D.~Lucanu, J.~Rutten, and A.~Silva.
\newblock A decision procedure for bisimilarity of generalized regular
  expressions.
\newblock In {\em Proceedings of the 13th Brazilian conference on Formal
  methods: foundations and applications}, SBMF'10, pages 226--241, Berlin,
  Heidelberg, 2011. Springer-Verlag.

\bibitem{Bouhoula-Jouannaud-Meseguer00}
A.~Bouhoula, J.-P. Jouannaud, and J.~Meseguer.
\newblock Specification and proof in membership equational logic.
\newblock {\em Theor. Comput. Sci.}, 236(1-2):35--132, 2000.

\bibitem{DBLP:conf/maude/2007}
M.~Clavel, F.~Dur\'{a}n, S.~Eker, P.~Lincoln, N.~Mart\'{\i}-Oliet, J.~Meseguer,
  and C.~Talcott.
\newblock {\em All about {M}aude - a high-performance logical framework: how to
  specify, program and verify systems in rewriting logic}.
\newblock Springer-Verlag, Berlin, Heidelberg, 2007.

\bibitem{goguen-lin-rosu-2000-ase}
J.~Goguen, K.~Lin, and G.~Rosu.
\newblock Circular coinductive rewriting.
\newblock In {\em ASE '00: Proceedings of the 15th IEEE international
  conference on Automated software engineering}, pages 123--132, Washington,
  DC, USA, 2000. IEEE Computer Society.

\bibitem{Goguen92order-sortedalgebra}
J.~A. Goguen.
\newblock Order-sorted algebra {I}: Equational deduction for multiple
  inheritance, overloading, exceptions and partial operations.
\newblock {\em Theoretical Computer Science}, 105:217--273, 1992.

\bibitem{acca}
E.-I. Goriac, D.~Lucanu, and G.~Ro\c{s}u.
\newblock Automating coinduction with case analysis.
\newblock In {\em Proceedings of the 12th international conference on Formal
  engineering methods and software engineering}, ICFEM'10, pages 220--236,
  Berlin, Heidelberg, 2010. Springer-Verlag.

\bibitem{DBLP:conf/fase/HausmannMS05}
D.~Hausmann, T.~Mossakowski, and L.~Schr{\"o}der.
\newblock Iterative {C}ircular {C}oinduction for {C}o{C}asl in
  {I}sabelle/{HOL}.
\newblock In M.~Cerioli, editor, {\em FASE}, volume 3442 of {\em Lecture Notes
  in Computer Science}, pages 341--356. Springer, 2005.

\bibitem{DBLP:journals/iandc/HermidaJ98}
C.~Hermida and B.~Jacobs.
\newblock Structural induction and coinduction in a fibrational setting.
\newblock {\em Inf. Comput.}, 145(2):107--152, 1998.

\bibitem{Kleene61}
S.~Kleene.
\newblock Representation of events in nerve nets and finite automata.
\newblock {\em Automata Studies}, pages 3--42, 1956.

\bibitem{kozen91}
D.~Kozen.
\newblock A completeness theorem for {K}leene algebras and the algebra of
  regular events.
\newblock In {\em LICS}, pages 214--225. IEEE Computer Society, 1991.

\bibitem{kozen-nerode}
D.~Kozen.
\newblock Myhill-{N}erode relations on automatic systems and the completeness
  of {K}leene algebra.
\newblock In A.~Ferreira and H.~Reichel, editors, {\em STACS}, volume 2010 of
  {\em Lecture Notes in Computer Science}, pages 27--38. Springer, 2001.

\bibitem{K08a}
D.~Kozen.
\newblock On the coalgebraic theory of {K}leene algebra with tests.
\newblock Technical Report~{\url{http://hdl.handle.net/1813/10173}}, Computing
  and Information Science, Cornell University, March 2008.

\bibitem{lucanu-etal-2009-calco}
D.~Lucanu, E.-I. Goriac, G.~Caltais, and G.~Ro\c{s}u.
\newblock {CIRC}: a behavioral verification tool based on circular coinduction.
\newblock In {\em Proceedings of the 3rd international conference on Algebra
  and coalgebra in computer science}, CALCO'09, pages 433--442, Berlin,
  Heidelberg, 2009. Springer-Verlag.

\bibitem{milner}
R.~Milner.
\newblock A complete inference system for a class of regular behaviours.
\newblock {\em J.~Comput.~System~Sci.}, 28(3):439--466, 1984.

\bibitem{derivatives-jfp09}
S.~Owens, J.~H. Reppy, and A.~Turon.
\newblock Regular-expression derivatives re-examined.
\newblock {\em J. Funct. Program.}, 19(2):173--190, 2009.

\bibitem{rosu-lucanu-2009-calco}
G.~Ro\c{s}u and D.~Lucanu.
\newblock Circular coinduction: a proof theoretical foundation.
\newblock In {\em Proceedings of the 3rd international conference on Algebra
  and coalgebra in computer science}, CALCO'09, pages 127--144, Berlin,
  Heidelberg, 2009. Springer-Verlag.

\bibitem{rosu-thesis}
G.~Rosu.
\newblock {\em Hidden Logic}.
\newblock PhD thesis, University of California at San Diego, 2000.

\bibitem{ccsl}
J.~Rothe, H.~Tews, and B.~Jacobs.
\newblock The coalgebraic class specification language {CCSL}.
\newblock {\em J. UCS}, 7(2):175--193, 2001.

\bibitem{Rutten00}
J.~J. M.~M. Rutten.
\newblock Universal coalgebra: a theory of systems.
\newblock {\em Theor. Comput. Sci.}, 249(1):3--80, 2000.

\bibitem{salomaa}
A.~Salomaa.
\newblock Two complete axiom systems for the algebra of regular events.
\newblock {\em J. ACM}, 13(1):158--169, 1966.

\bibitem{bbrs_ic}
A.~Silva, F.~Bonchi, M.~Bonsangue, and J.~Rutten.
\newblock Quantitative {K}leene coalgebras.
\newblock {\em Information and Computation}, 209(5):822--849, 2011.

\bibitem{brs_lmcs}
A.~Silva, M.~M. Bonsangue, and J.~J. M.~M. Rutten.
\newblock Non-deterministic {K}leene coalgebras.
\newblock {\em Logical Methods in Computer Science}, 6(3), 2010.

\bibitem{staton}
S.~Staton.
\newblock Relating coalgebraic notions of bisimulation: with applications to
  name-passing process calculi.
\newblock In {\em Proceedings of the 3rd international conference on Algebra
  and coalgebra in computer science}, CALCO'09, pages 191--205, Berlin,
  Heidelberg, 2009. Springer-Verlag.

\bibitem{worthington}
J.~Worthington.
\newblock Automatic proof generation in {K}leene algebra.
\newblock In R.~Berghammer, B.~M{\"o}ller, and G.~Struth, editors, {\em
  RelMiCS}, volume 4988 of {\em Lecture Notes in Computer Science}, pages
  382--396. Springer, 2008.

\end{thebibliography}

\end{document}